\documentclass{scrarticle}

\usepackage[utf8]{inputenc}
\usepackage[T1]{fontenc}
\usepackage[UKenglish]{babel}
\usepackage{authblk}
\usepackage[a4paper,margin=1in]{geometry}
\usepackage{amsmath,amsfonts,amsthm,amssymb}
\usepackage[noabbrev,capitalize]{cleveref}
\usepackage{tikz}
\usetikzlibrary{patterns,decorations.pathreplacing,calligraphy}
\usepackage{csquotes}
\usepackage{nicefrac}
\usepackage{mathtools}
\usepackage{natbib}

\theoremstyle{plain}
\newtheorem{theorem}{Theorem}
\newtheorem{lemma}[theorem]{Lemma}
\newtheorem{corollary}[theorem]{Corollary}
\theoremstyle{definition}
\newtheorem{example}[theorem]{Example}
\newtheorem{remark}[theorem]{Remark}
\newtheorem{definition}{Definition}
\newtheorem*{ack}{Acknowledgment}
\newtheorem*{funding}{Funding}

\newcommand{\calA}{\ensuremath{\mathcal{A}}}
\newcommand{\calB}{\ensuremath{\mathcal{B}}}
\newcommand{\calF}{\ensuremath{\mathcal{F}}}
\newcommand{\calI}{\ensuremath{\mathcal{I}}}
\newcommand{\calL}{\ensuremath{\mathcal{L}}}
\newcommand{\calP}{\ensuremath{\mathcal{P}}}
\newcommand{\calQ}{\ensuremath{\mathcal{Q}}}
\newcommand{\N}{\ensuremath{\mathbb{N}}}
\newcommand{\QQ}{\ensuremath{\mathbb{Q}}}
\newcommand{\RR}{\ensuremath{\mathbb{R}}}
\newcommand{\concake}{[0,1]}
\newcommand{\abscake}{\ensuremath{X}}
\newcommand{\LIM}{\mathop{\operatorname{LIM}}}
\newcommand{\citeA}[2][]{\citet[#1]{#2}}
\newcommand{\citeR}[2][]{\citealp[#1]{#2}}
\renewcommand{\S}{§}

\setkomafont{author}{\sffamily}

\begin{document}
\title{Cutting a Cake Is Not Always a \enquote{Piece of Cake}: A
  Closer Look at the Foundations of Cake-Cutting Through the Lens of
  Measure Theory}

\author[1]{Peter Kern}
\author[2]{Daniel Neugebauer}
\author[2]{Jörg Rothe}
\author[3]{René L.\ Schilling}
\author[4]{Dietrich Stoyan}
\author[2,*]{Robin Weishaupt}

\affil[1]{Mathematisches Institut, Heinrich-Heine-Universität Düsseldorf}
\affil[2]{Institut für Informatik, Heinrich-Heine-Universität Düsseldorf}
\affil[3]{Institut für Mathematische Stochastik, TU Dresden}
\affil[4]{Institut für Stochastik, TU Bergakademie Freiberg}
\affil[*]{Corresponding Author: robin.weishaupt@hhu.de}

\date{\vspace{-5ex}}

\maketitle

\begin{abstract}
  Cake-cutting is a playful name for the fair division of a
  heterogeneous, divisible good among agents, a well-studied problem at
  the intersection of mathematics, economics, and artificial
  intelligence.  The cake-cutting literature is rich and edifying.
  However, different model assumptions are made in its many papers, in
  particular regarding the set of allowed pieces of cake that are to be
  distributed among the agents and regarding the agents' valuation
  functions by which they measure these pieces.  We survey the commonly
  used definitions in the cake-cutting literature, highlight their
  strengths and weaknesses, and make some recommendations on what
  definitions could be most reasonably used when looking through the
  lens of measure theory.\\

  \noindent\textbf{Keywords}: cake-cutting protocol; admissible piece of cake; finitely
  additive measure; continuity properties of measures
\end{abstract}

\section{Introduction}
\label{chap:introduction}

Since the groundbreaking work of \citeA{ste:j:fair-division},
cake-cutting is a metaphor for the so-called \emph{fair division
  problem for a divisible, heterogeneous good}, which addresses the
problem to split a contested quantity (a \enquote{cake}) in a fair way
among several parties $A,B,C,\dots$; each party may have its own idea
about the value of the different parts of the cake.

While mainly mathematicians and economists were concerned with the
study of cake-cutting early on, \emph{``in recent years, cake cutting
  has emerged as a major research topic in artificial intelligence,''}
as \citeA[p.~567]{bal:c:simultaneous-cake-cutting} note.  They
substantiate their claim by listing ten papers on cake-cutting five of
which appeared in AAAI (e.g.,
\cite{coh-lai-par-pro:c:optimal-envy-free-cake-cutting}), three in
IJCAI (e.g., \cite{pro:c:though-shalt-covet}), and the remaining two
in AAMAS proceedings (e.g.,
\cite{aum-dom-has:c:socially-efficient-cake-divisions}).  For more
than a decade now, AAAI and IJCAI (the two top AI conferences) and
AAMAS (the leading venue for research on multiagent systems) have
published numerous research papers on fair division and, in
particular, on cake-cutting.
\citeA[p.~567]{bal:c:simultaneous-cake-cutting} go on to write,
\emph{``The growing interest in cake cutting, and fair division more
  broadly, is partly motivated by potential applications in AI, such
  as industrial procurement, manufacturing and scheduling, and airport
  traffic management
  \citep{che-etal:j:multiagent-resource-allocation}.  For example,
  concrete applications to the allocation of multiple computational
  resources in shared computing systems have recently received
  significant attention
  \citep{gut-nis:c:fair-allocation-without-trade,kas-pro-sha:c:no-agent-left-behind}.''}
The main purpose of our paper, however, are neither applications nor
novel protocols for cake-cutting; instead, we will have a closer look
at the mathematical foundations of cake-cutting, establishing the
connection to measure theory. Under very modest assumptions (e.g., the
possibility to make continuous cuts) this will empower researchers in
the field with new tools which are potent enough to deal with
situations that are currently seen as ``exotic'' or ``theoretical.''

A traditional way
of fair division between two parties $A$ and $B$ would be to let $A$
divide the cake into two pieces (depending on their own valuation)
while $B$ has the right to choose one of the pieces, the so-called
\emph{cut \& choose} protocol. There are other possibilities for two
parties as well as extensions to more than two parties (see, e.g.,
\citeR{pro:b:handbook-comsoc-cake-cutting-algorithms,
lin-rot:b:economics-and-computation-cake-cutting}, for an overview).
Yet, while the basic rules of the game are pretty clear, the
assumptions on the actual cutting process are often treated in a
gentlemanlike manner. If the whole cake is represented by an interval,
say $\concake$, many authors think of the pieces as \enquote{intervals,}
without specifying whether the intervals are open $(a,b)\subset
\concake$, half-open $(a,b], [a,b)\subset \concake$, or closed
$[a,b]\subseteq\concake$, and how to treat the~-- possibly twice
counted~-- end points, i.e., $[0,\nicefrac{1}{2}) \cup
[\nicefrac{1}{2},1]$ vs.\ $[0,\nicefrac{1}{2}] \cup
[\nicefrac{1}{2},1]$; this is, of course, not an issue if a one-point
set like $\{\nicefrac{1}{2}\}$ has zero value for all
parties. However, this simple example shows that a formal mathematical
approach to cake-cutting needs to address questions like:
\begin{itemize}
\item Are (open, closed, half-open) intervals the only possible pieces
  of cake?
\item Do we allow for finitely many or infinitely many cuts?  A
  \enquote{cut} means the split of any subset of $\concake$ at a single
  point; it depends on the particular protocol to which interval the
  point will belong.
\item Which properties should a valuation function (by which an agent
  individually evaluates the pieces of cake) have, and how does it
  interact with the family of admissible pieces of cake?
\end{itemize}
For some cases, there is an obvious answer: If we use only finitely many
cuts, finite unions of intervals of the form $\langle a,b\rangle$~--
where the angular braces indicate either open or closed ends~-- is all we
can get; and if, in addition, any single point $a\in \concake$ has zero
value, we do not have to care about the open or closed ends anymore.
We will see in Section~\ref{chap:preliminaries:finite} below that this
rather implicit assumption brings us in a much more potent framework
that can effectively deal with a countably infinite number of cuts.

Let us briefly discuss situations where an infinite number of cuts may
actually be inevitable.  While most research in cake-cutting has
focused on finite protocols and on minimizing the required number of
cuts, there are also some impossibility results that show that no
finite cake-cutting protocol (even if unbounded) can guarantee all
players their fair share of the cake (for various notions of fairness
such as proportionality, exactness, or envy-freeness).  For example,
\citeA{str:j:finite-protocols-cannot-ef} shows that no finite
cake-cutting protocol can guarantee an envy-free division of a cake
among three or more players who each are to receive a single connected
piece.  As another example, in contrast to the moving-knife procedure
due to \citeA{aus:j:moving-knife-cut-and-choose} that guarantees two
players an exactly proportional share,
\citeA{rob-web:b:cake-cutting-algorithms-be-fair-if-you-can} show that
no finite cake-cutting protocol, bounded or unbounded, can guarantee
an exactly proportional division of the cake for two players; see also
the related nearly exact, envy-free, finite unbounded protocol by
\citeA{rob-web:j:near-exact-and-envy-free-cake-division}.
Such impossibility results indicate that infinite cake-cutting protocols
are unavoidable when one has to deal with certain valuations of the
cake (which, admittedly, are usually constructed specifically for the
purpose of proving the desired impossibility result).

As soon as we allow for countably infinitely many cuts, things change
dramatically, as the following example shows.

\begin{example}[Cantor dust; Cantor's ternary set]\label{ex:cantor-set}
  Start with the complete cake as a single piece, i.e., $A_0 =
  \concake$. Now, cut out the middle third of $A_0$ to obtain the
  intermediate piece $A_1 = A_0\setminus (\nicefrac{1}{3},
  \nicefrac{2}{3}) = [0,\nicefrac{1}{3}] \cup [\nicefrac{2}{3},1]$
  comprising two closed intervals. Next, cut out the middle third of
  both remaining pieces in $A_1$ to obtain a union of four closed
  intervals $A_2 = [0,\nicefrac{1}{9}] \cup
  [\nicefrac{2}{9},\nicefrac{1}{3}] \cup [\nicefrac{2}{3},
  \nicefrac{7}{9}] \cup [\nicefrac{8}{9}, 1]$, see
  \cref{fig:motivation:cantor-construction-sets-1}.
  If this procedure is
  repeated on and on, we will remove countably many open intervals, and
  the remainder set is $C_{\nicefrac{1}{3}} = \bigcap_{i = 1}^{\infty}
  A_i$. The set $C_{\nicefrac{1}{3}}$ is the \textbf{Cantor (ternary)
    set} (see, e.g.,~\citeR[\S~2.5]{schi:b:counterexamples}), and one can
  show
  that this is a closed set, which has more than countably many points,
  does not contain any interval, and
  is dense in itself, i.e., each
  of its points is a limit point of a sequence inside
  $C_{\nicefrac{1}{3}}$. In the usual measuring scale, the original cake
  had length~$1$, and the recursively removed pieces have
  total length
  \begin{gather*}
    \frac 13 + \left(\frac 19 + \frac 19\right) + \left(\frac 1{27} +
      \frac 1{27}+\frac 1{27} + \frac 1{27}\right) + \cdots
    = \sum_{i\in\N} \frac{2^{i-1}}{3^i}
    = 1,
  \end{gather*}
  so that $C_{\nicefrac{1}{3}}$ has zero \enquote{length,} but it still
  contains more than countably many points.

  The same construction principle, removing at each stage $2^{i-1}$
  identical open middle intervals, each having length $p^i$ for some~$p$,
  $0<p\leq \nicefrac{1}{3}$, leads to the Cantor set $C_p$, which is,
  again, closed, uncountable, and does not contain any interval. If,
  say, $p=\nicefrac{1}{4}$, the removed intervals have total
  length $\nicefrac{1}{2}$ and the remaining Cantor dust has
  \enquote{length} $1- \nicefrac{1}{2} = \nicefrac{1}{2}$. This is not
  quite expected.

  While it is intuitive that the removed intervals should have a
  certain length, it feels unnatural to speak of the \enquote{length} of
  a dust-like set as $C_p$. In fact, we are dealing here with
  (one-dimensional) Lebesgue measure, which is the mathematically formal
  extension of the familiar notion of \enquote{length.}
\end{example}

An alternative, slightly more formal way of illustrating the Cantor
dust is given in the appendix as Example~\ref{ex:cantor-set-appendix}.

This example shows that, as soon as we allow for countably many cuts,
there can appear sets which may not be written as a countable union of
intervals; moreover, although these sets consist of limit points only,
they may have strictly positive length.
\begin{figure}[t]
  \centering
  \includegraphics[width=\textwidth]{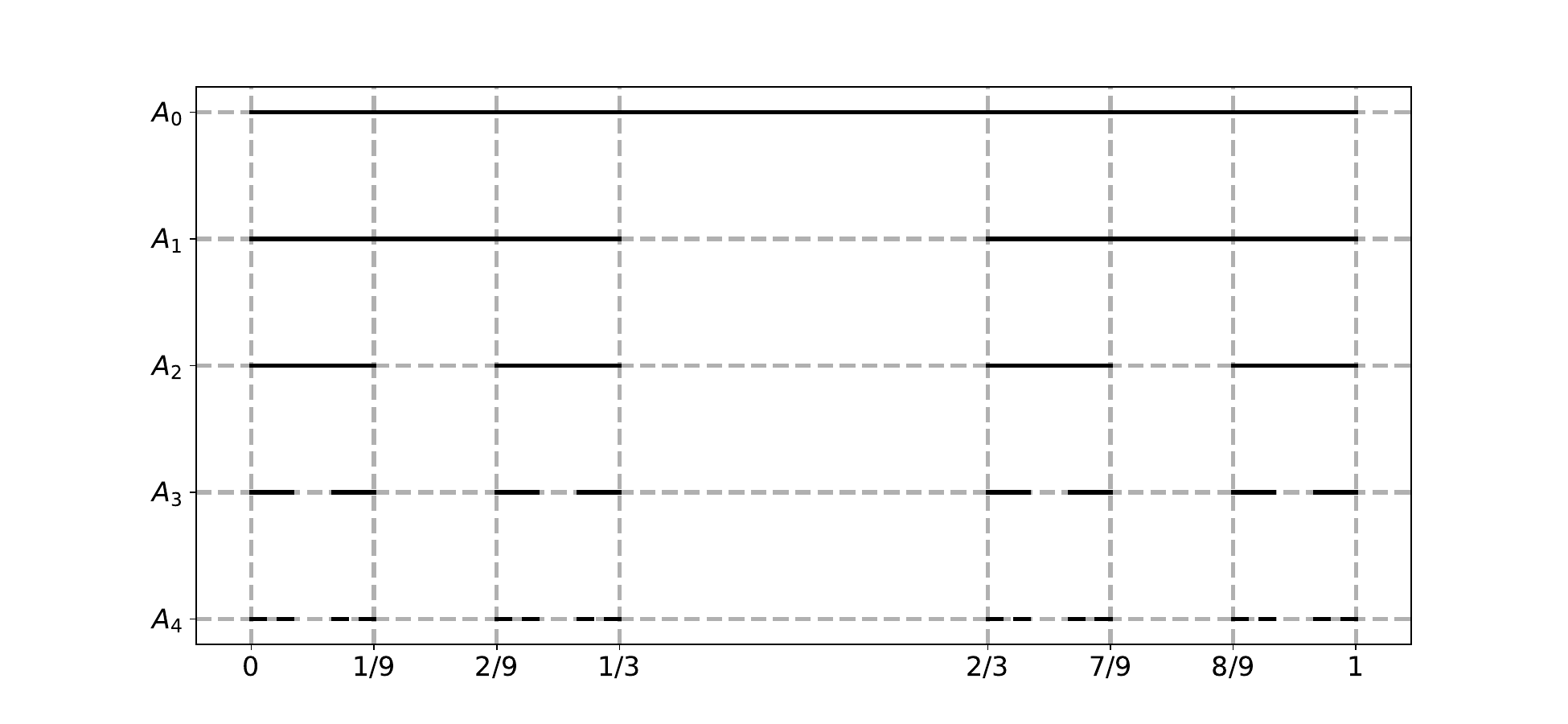}
  \caption{Step-wise pieces to be cut for a Cantor-like piece of
    cake.}
  \label{fig:motivation:cantor-construction-sets-1}
\end{figure}

An important feature of this example is the fact that we extend the
family of intervals to a family of subsets which contains (i)~finite
unions, (ii)~countable intersections, and (iii)~complements of its
members, leading to fairly complicated subsets as, e.g.,
$C_p$. Moreover, when calculating the length of all removed intervals,
we tacitly assumed
\begin{itemize}
\item the (finite) additivity of length: The length of two disjoint
  sets is the sum of their lengths;
\item the countable or $\sigma$-additivity which plays the role of a
  continuity property: The length of a countable union is the limit of
  the length of the union of the first $N$ sets as $N \to \infty$.
\end{itemize}
As it will turn out, these are two far-reaching assumptions on the
interplay of the valuation function (here: length) with its domain; we
will see how this relates to the desirable property that we can cut
off pieces of arbitrary length~$\ell$, $0\leq \ell\leq 1$, from the
cake $\concake$ (allowing for any valuation values of the cut-off
pieces).

Commonly, in cake-cutting theory (see,
e.g.,~\citeR{bra-tay:b:fair-division-from-cake-cutting-to-dispute-resolution,
pro:b:handbook-comsoc-cake-cutting-algorithms,
lin-rot:b:economics-and-computation-cake-cutting}) a (piece-wise
constant) valuation function $v \colon \calP \to [0,1]$, where $\calP$
is some family of subsets of the cake, is represented as shown in
\cref{fig:motivation:common-valuation-function}: The cake $\concake$
is split horizontally into multiple pieces and the number of
vertically stacked boxes per piece describes the piece's valuation
from some agent's perspective. For example, the valuation function $v$
in Figure \ref{fig:motivation:common-valuation-function} evaluates the
piece $X' = [0,\nicefrac{2}{6}]$ with $v(X') = \nicefrac{3}{17}$.
\begin{figure}[ht]
    \centering
    \begin{tikzpicture}[scale=0.7]
      \foreach \x in {1,2,...,6}
      \draw (-0.5+\x,-0.5) node{\x};
      \draw[very thick] (0,0) grid  (1,2);
      \draw[very thick] (1,0) grid +(1,1);
      \draw[very thick] (2,0) grid +(1,5);
      \draw[very thick] (3,0) grid +(1,2);
      \draw[very thick] (4,0) grid +(1,4);
      \draw[very thick] (5,0) grid +(1,3);
    \end{tikzpicture}
    \caption{Common representation for a valuation function in
      cake-cutting.}
    \label{fig:motivation:common-valuation-function}
\end{figure}
Having this example in mind, one is not aware of any limitations and
might assume that all possible sets are indeed admissible pieces,
i.e., $\calP = \mathfrak{P}(\concake) = \{A \mid
A\subseteq\concake\}$.

The following classical example from measure theory shows that there
cannot exist a valuation function that assigns to intervals $\langle
a,b\rangle\subseteq \concake$ their natural length $b-a$, and which is
additive, $\sigma$-additive (in the sense explained above), and able
to assign a value to every set $A\subseteq \concake$. Things are
different if we do not require $\sigma$-additivity (see the discussion
in~\citeR[\S~7.31]{schi:b:counterexamples}).

\begin{example}[\citeR{vit:c:vitali-sets};
  see also, e.g.,
  \citeR{schi:b:counterexamples}]\label{ex-vitali}
  Let $\concake$ be the standard cake, and assume that the
  valuation function $v$ is $\sigma$-additive (see
  Definition~\ref{def:valuation-function} on
  page~\pageref{def:valuation-function}), assigning to any interval
  its natural length. This means, in particular, that $v$ is invariant
  under translations and evaluates the complete cake with $v(\concake) = 1$.
  Let us define the relation $\ast$ as follows: We say that two real
  numbers $x,y \in \RR$ satisfy the relation $\ast$ if, and only if, $x
  - y \in \QQ$, i.e., their difference is rational.  The relation $\ast$
  is an equivalence relation and the corresponding equivalence classes
  $[x] = \{ y \in \RR \mid x \ast y \} \subseteq \RR$ lead to a disjoint
  partitioning of $\RR$.  By the axiom of choice, there is a set $V
  \subset \concake$ which contains exactly one representative of every
  equivalence class~$[x]$.  A set like $V$ is called a \textbf{Vitali
    set}.  Clearly, $V \in \mathfrak{P}(\concake)$ and the sets $q+V =
  \{q+x\mod 1 \mid x\in V\}$, $q\in\QQ$, are a disjoint partition of
  $\concake$; thus $\bigcup_{q\in\QQ} (q+V) = \concake$. By assumption, $v$ is
  $\sigma$-additive and assigns to each $q+V$ the same value
  (translation invariance).  Hence, we end up with the contradiction
  \begin{gather*}
    1
    = v(\concake)
    = v\left(\bigcup_{q\in\QQ}(q+V)\right)
    = \sum_{q\in\QQ} v(q+V)
    = \begin{cases}0 & \text{ if }v(V)=0,\\ \infty & \text{ if }v(V)>0.\end{cases}
  \end{gather*}
  Thus $v$ cannot have the power set of the cake $\concake$ as its
  domain \emph{if we assume that $v$ is $\sigma$-additive}. We will see
  below that certain commonly used divisibility assumptions are
  equivalent to the $\sigma$-additivity of the valuation.
\end{example}

\begin{example}[Cantor function]\label{ex:cantor-function}
  Let us return to Example~\ref{ex:cantor-set} and interpret the
  points in the set $C_{\nicefrac{1}{3}}$ as valuable assets which
  need to be priced. We may assume that the total value of the cake
  $C_{\nicefrac{1}{3}}$ is $1$. We want to construct a
  \enquote{cumulative valuation function~$V$} which has the property
  that for $0\leq a \leq b\leq 1$ the difference $V(b)-V(a)$ is the
  value of the points contained in $C_{\nicefrac{1}{3}}\cap (a,b]$.
  Clearly, $x\mapsto V(x)$ is a (not necessarily strictly) increasing
  function with $V(0)=0$ and $V(1)=1$.

  If we agree that the assets should be \enquote{homogeneously}
  priced, then we are automatically led to the following scheme: As
  the total value of $C_{\nicefrac{1}{3}}$ is one, the value of
  $C_{\nicefrac{1}{3}}\cap [0,\nicefrac 12]$ and
  $C_{\nicefrac{1}{3}}\cap [\nicefrac 12, 1]$ should be the same,
  i.e., $\nicefrac 12$.  Since $C_{\nicefrac{1}{3}} \cap (\nicefrac
  13, \nicefrac 23) = \emptyset$,
  we see that both the left third and the right third of
  $C_{\nicefrac{1}{3}}$ has the value $\nicefrac 12$.  This means that
  $V(x) = \nicefrac 12$ on the whole middle third $(\nicefrac
  13,\nicefrac 23)$.

  Now we can repeat this argument in the two remaining sets
  $C_{\nicefrac{1}{3}}\cap [0,\nicefrac 13]$ and
  $C_{\nicefrac{1}{3}}\cap [\nicefrac 23,1]$. Since these pieces are
  scaled-down versions of the original set $C_{\nicefrac{1}{3}}$, we can
  repeat our argument to the three thirds of the scaled sets and so we
  see that the cumulative valuation function $V(x)$ takes the values
  $\nicefrac 14$ and $\nicefrac 34$ on the intervals $(\nicefrac
  19,\nicefrac 29)$ and $(\nicefrac 79,\nicefrac 89)$, respectively.

  Iterating this procedure \emph{ad infinitum}, the remaining values
  of $V$ at interfaces of the intervals are uniquely determined by
  monotonicity and we end up with the so-called \emph{Cantor function}
  or \emph{devil's staircase}, which is monotone, increasing,
  continuous, and it is flat (i.e., constant) on all middle-thirds
  removed in the construction process of $C_{\nicefrac{1}{3}}$ in
  Example~\ref{ex:cantor-set}; a precise mathematical description can be
  achieved, e.g., using the alternative representation of the Cantor set
  in Example~\ref{ex:cantor-set-appendix} of
  Appendix~\ref{sec:appendix}, but at this point the pictures in
  \cref{fig:motivation:cantor-construction-sets} tell it all:
  \begin{figure}[t]
    \centering
    \begin{tikzpicture}
      \def \lightGray {gray!60};
      \draw [black, step=1.0cm] (0,0) grid +(3,3);
      \draw[black, pattern=north east lines, pattern color=\lightGray] (1,0) rectangle (2,1);
      \draw[black, pattern=north east lines, pattern color=\lightGray] (2,0) rectangle (3,1);
      \draw[black, pattern=north east lines, pattern color=\lightGray] (2,1) rectangle (3,2);
      \draw[black, pattern=north east lines, pattern color=\lightGray] (0,1) rectangle (1,2);
      \draw[black, pattern=north east lines, pattern color=\lightGray] (0,2) rectangle (1,3);
      \draw[black, pattern=north east lines, pattern color=\lightGray] (1,2) rectangle (2,3);
      \draw [black, thin] (0,0) -- (3,3);
      \draw [black, semithick] (1,1.5) -- (2,1.5);
      \draw [decorate,decoration = {calligraphic brace}] (3,-0.1) --
      (0,-0.1) node[pos=0.5,below=2pt,black]{\tiny Cake};
      \draw [decorate,decoration = {calligraphic brace}] (-0.1,0) --
      (-0.1,3) node[pos=0.5,left=1pt,black]{\tiny $V(x)$};

      \draw [black] (4,0) rectangle (7,3);
      \draw [black, densely dashed] (5,1) rectangle (6,2);
      \draw [black] (6,2) rectangle (7,3);
      \draw [black, thin] (4,0) -- (7,3);
      \draw [black, semithick] (5,1.5) -- (6,1.5);

      \draw [black] (4,0) rectangle (5,1);
      \draw [black] (4,0) rectangle (4.33,0.33);
      \draw [black, densely dashed] (4.33,0.33) rectangle (4.66,0.66);
      \draw [black] (4.66,0.66) rectangle (5,1);
      \draw [black, semithick] (4.33,0.495) -- (4.66,0.495);

      \draw [black, step=0.333333333333cm] (6,2) grid +(1,1);
      \draw[black, pattern=north east lines, pattern color=\lightGray] (6.33333,2) rectangle (6.66666,2.33333);
      \draw[black, pattern=north east lines, pattern color=\lightGray] (6.66666,2) rectangle (7,2.33333);
      \draw[black, pattern=north east lines, pattern color=\lightGray] (6.66666,2.33333) rectangle (7,2.666666);
      \draw[black, pattern=north east lines, pattern color=\lightGray] (6,2.33333) rectangle (6.33333,2.666666);
      \draw[black, pattern=north east lines, pattern color=\lightGray] (6,2.66666) rectangle (6.33333,3);
      \draw[black, pattern=north east lines, pattern color=\lightGray] (6.33333,2.66666) rectangle (6.666666,3);
      \draw [black, semithick] (6.33333,2.495) -- (6.66666,2.495);

      \node (img1) at (9.5,1.5) {\includegraphics[width=3cm,height=3cm]{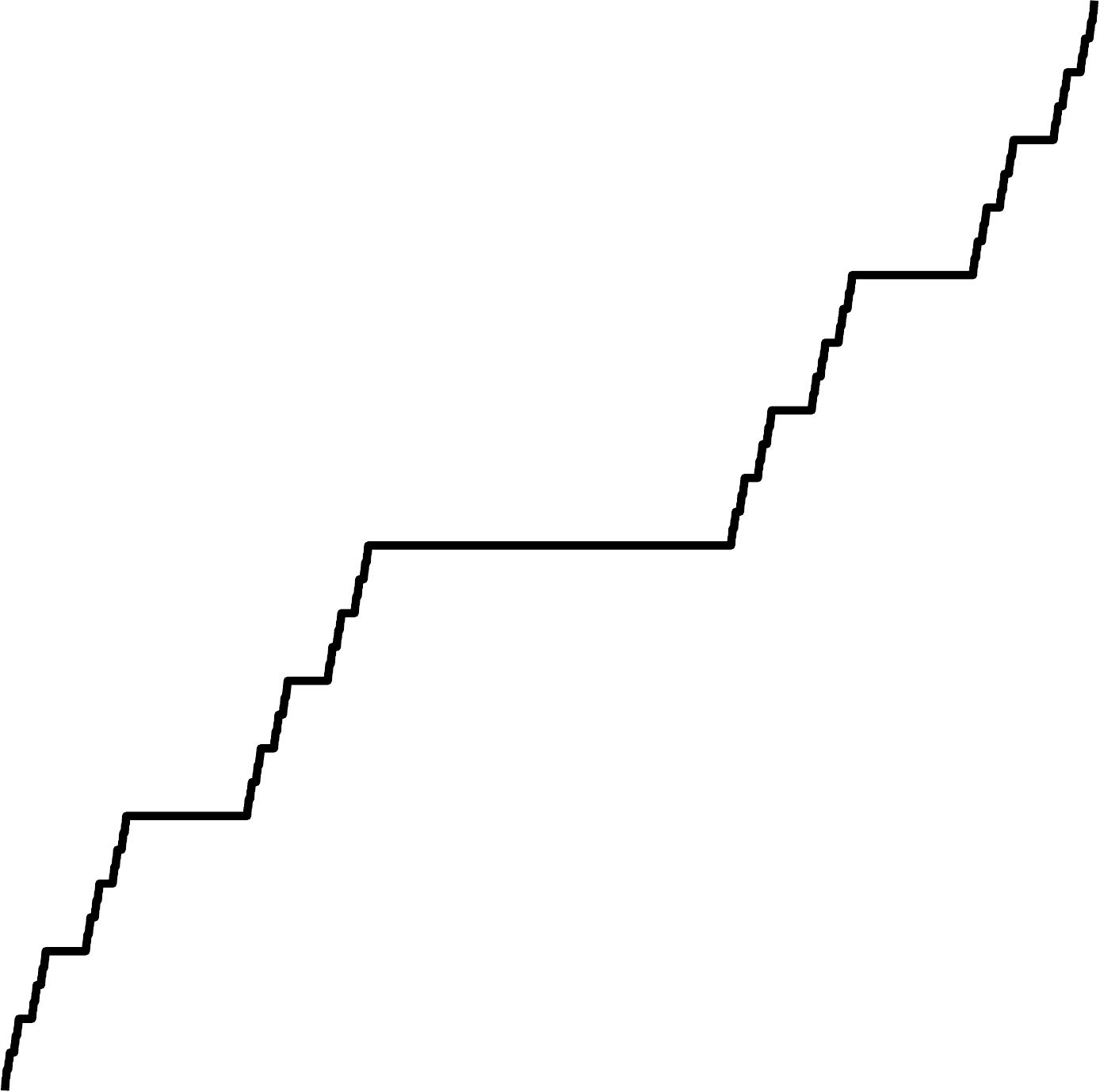}};
      \draw [black] (8,0) rectangle (11,3);
    \end{tikzpicture}
    \caption{Step-by-step construction of the Cantor function: In
      each step, we subdivide the top right and bottom left squares into
      nine smaller squares. We keep only the three squares along the
      diagonal, and discard (gray out) the off-diagonal squares. The middle
      square is halved by a horizontal line (here $V$ is constant). Repeat.}
    \label{fig:motivation:cantor-construction-sets}
  \end{figure}
  The Cantor function is the typical function where the fundamental
  theorem of integral and differential calculus fails: $V(x)$ is
  constant on all intervals contained in $\concake\setminus C_{\nicefrac
    13}$. Since $V$ is increasing, we can set $V'(x) := \limsup_{h\to 0}
  \frac 1h (V(x+h)-V(x))$, and this is the usual derivative whenever it
  exists. In particular, $V'(x)=0$ on $\concake\setminus C_{\nicefrac
    13}$. On the other hand, we have
  \begin{gather*}
    \int_0^x V'(t)\,dt = \int_{C_{\nicefrac 13}\cap [0,x)} V'(t)\,dt = 0 \neq V(x)-V(0)\quad\text{for any $x\in(0,1]$}.
  \end{gather*}
  This happens because the \enquote{length} of $C_{\nicefrac 13}$ is
  zero, i.e., we integrate over \enquote{too small a set} (no matter
  how big the integrand may be, it could even take the value
  $+\infty$!) so as to pick up any strictly positive value.
\end{example}

The above examples highlight some of the problems when evaluating
sets. A Cantor-like piece can only be evaluated if the valuation
function is not too simplistic. On the other hand, a Vitali set cannot
be evaluated at all if we request too many properties of a valuation
function, i.e., the domain $\mathfrak{P}(\concake)$ consisting of all
possible pieces of cake is, in general, too large.

Across the research field of cake-cutting (see, e.g., the textbooks
by~\citeR{bra-tay:b:fair-division-from-cake-cutting-to-dispute-resolution,
rob-web:b:cake-cutting-algorithms-be-fair-if-you-can}, and the book
chapters by~\citeR{pro:b:handbook-comsoc-cake-cutting-algorithms,
lin-rot:b:economics-and-computation-cake-cutting}), there exist
several different assumptions on the underlying model. Our goal is to
review thoroughly and comprehensively all the different models that
are currently applied in the literature. Furthermore, we study the
relationships between these models and formulate some related results.
It turns out that some of these models are problematic and should not
be used as they are formulated.  We highlight these models' problems
and provide specific examples showing why they are problematic.  Our
overall goal is to determine a model, which is as simple as possible,
yet powerful enough to cope with these problems and still compatible
with many of the currently used models.

Frequently, authors proposing cake-cutting protocols abstain from
making formal assumptions or from formalizing their model in detail.
For example, \citeA[p.~553]{bra-tay-zwi:c:moving-knife} write:
\begin{quote}
  \emph{\enquote{Many feel that the informality adds to the subject’s
      simplicity and charm, and we would concur.  But charm and simplicity
      are not the only factors determining the direction in which
      mathematics moves or should move. Our analysis in this paper raises
      several issues that may only admit a resolution via some negative
      results.  While such results may not require complete formalization of
      what is permissible, they do appear to require partial versions.  We
      will refer to such partial limitations as theses.}}
\end{quote}

It would thus be desirable to have some common consensus on which
models are useful for any given purpose, and which are not. If we
allow only a fixed number of cuts, splitting the cake $\concake$ into a
finite number of pieces of the type $\langle a,b\rangle \subseteq
\concake$, a naive approach is always possible: The valuation should be
additive and its domain contains unions of finitely many
intervals. If, on the other hand, there are potentially infinitely
many cuts~-- e.g., if the players play a game resulting in an
\emph{a priori} not fixed number of rounds
(such as the finite unbounded envy-free cake-cutting protocol of
\citeA{bra-tay:j:protocol})~--
the limiting case cannot any
longer be treated by a finitely additive valuation and a domain
containing only finite unions, see Example~\ref{ex:cantor-set}.

We propose to use ideas from measure theory, which provides the right
toolbox to tackle the issues described above. We will see that, at
least for the cake $\concake$, even the naive approach plus the
requirement that we can split every piece $\langle a,b\rangle$ by a
single cut into any proportion (in fact, a slightly weaker requirement
will do, cf.\ Definition~\ref{def:valuation-function}~$\mathrm{(D)}$),
automatically leads to the measure-theoretic point of view. That is to
say that in many natural situations the naive standpoint is
\enquote{practically safe} since its obvious shortcomings are
automatically \enquote{fixed by (measure) theory,} if one uses the
correct formulation.

\section{The Rules of the Game}
\label{chap:preliminaries}

Throughout this paper, $\concake$ denotes a standard cake, and the
power set $\mathfrak{P}(\concake) = \{S \mid S\subseteq \concake\}$
are all \textbf{possible} pieces of cake from a set-theoretic point of
view.  We define $\calP \subseteq \mathfrak{P}(\concake)$ as the set
of all \textbf{admissible} pieces of~$\concake$, i.e., those pieces
which (a)~can be allocated to some players via a cake-cutting
protocol, and (b)~can be evaluated by the players using their
valuation functions.  Sometimes it is necessary to consider an
\enquote{abstract} cake $\abscake$, with its possible and admissible
pieces $\mathfrak{P}(\abscake)$ and $\calP\subseteq
\mathfrak{P}(\abscake)$.  Some results for the standard cake
$\concake$ remain true for abstract cakes. For example, an abstract
cake $\abscake$ might be contained in the $n$-dimensional unit cube:
$X \subseteq [0,1]^n$.

\subsection{Dividing a Cake with Finitely Many Cuts}
\label{chap:preliminaries:finite}

We start by formulating requirements for $\calP$ regarding the
admissible pieces of cake.  The discussion in this section
applies both to
the standard cake $\concake$ and the abstract
cake $\abscake$.  Obviously, we want to be able to allocate
the complete cake $\abscake$ as well as an empty piece $\emptyset$ to a
player and therefore, $\abscake \in \calP$ and $\emptyset \in \calP$ must
hold. If $A \subseteq \abscake$ is already allocated to some player,
i.e., $A \in \calP$, then we want to be able to give the remainder of
the cake to another player; so for all $A \in \calP$, we demand that
the complement of~$A$, denoted by $\overline{A} = \abscake \setminus A$,
is in~$\calP$.
Furthermore, we want to be able to cut and combine pieces of cake; so
for all $A, B \in \calP$, we require $A \cup B \in \calP$.  Note that
$A \cap B = \overline{\overline{A} \cup \overline{B}}$ and $A
\setminus B = A \cap \overline{B}$, so our previously formulated
requirements also allow us to allocate the intersection of a finite
number of pieces of cake and to evaluate the difference of two pieces
of cake.

\begin{definition}\label{def:algebra}
Let $\abscake$ be a(n abstract) cake.  A family $\calA
\subseteq \mathfrak{P}(\abscake)$ is called an \textbf{algebra} over
$\abscake$ if $\emptyset \in \calA$ and for all $A,B\in\calA$ it holds
that $\overline{A}$ and $A \cup B \in \calA$.
\end{definition}

It is worth noting that only by the formulation of intuitive
requirements with respect to the set of all admissible pieces of cake,
we ended up with a well-studied, structured concept from measure
theory: an algebra.

\begin{example}
  If $\abscake = \concake$, then $\mathfrak{P}(\abscake)$ and
  $\{\emptyset,\abscake\}$ are algebras~-- in fact these are the largest
  possible and the smallest possible algebras over $\concake$. Another
  useful algebra is the family $\calI(\concake)$\label{pg:calI-cake} of
  all unions of finitely many intervals in $\concake$~-- and it is easy
  to check that $\calI(\concake)$ is the smallest algebra containing all
  closed (or all open or all half-open) intervals from $\concake$. While
  it is obvious that $\{\emptyset,\concake\}$ is useless for our
  purpose, as then only two possible pieces can be allocated, the
  complete cake and an empty piece, we might~-- at the other extreme~--
  also take $\mathfrak{P}(\concake)$ as the set for the admissible
  pieces of $\concake$. However, when choosing $\calP$, we must also
  ensure that meaningful valuation functions can exist for this set, and
  Example~\ref{ex-vitali} shows that for a rather natural valuation
  function~-- geometric length~-- $\mathfrak{P}(\concake)$ is too big.
\end{example}

Let us list the common requirements for the players' valuation
functions.  A \textbf{valuation function} $v$ shall assign to any
admissible piece of cake $A \in \calP$ some positive real number.  In
order to normalize the players' valuations and keep them comparable,
we map the positive real numbers onto $[0,1]$ continuously,
bijectively, and preserving the natural order.  Hence, we can further
limit the valuation function's range to $[0,1]$, i.e., we have $v
\colon \calP \to [0,1]$.  The next definition lists desirable
properties for a valuation function.

\begin{definition}\label{def:valuation-function}
  Let $\abscake$ be a(n abstract) cake and $\calA$ the algebra of
  admissible pieces. A \textbf{valuation function} is a function $v
  \colon \calA \to [0,1]$, which is normalized, i.e., $v(\emptyset)=0$
  and $v(\abscake) = 1$.
  Moreover, $v$ is called
  \begin{enumerate}
  \item[$\mathrm{(M)}$] \textbf{monotone} if for $A,B\in\calA$ with
    $A\subseteq B$, one has $v(A)\leq v(B)$;
  \item[$\mathrm{(A)}$] \textbf{additive} or \textbf{finitely
      additive} if for all $A,B \in \calA$ such that $A\cap B=\emptyset$,
    one has $v(A \cup B) = v(A) + v(B)$;
  \item[$(\Sigma)$] $\boldsymbol{\sigma}$-\textbf{additive} or
    \textbf{countably additive} if for any sequence $(A_n)_{n\in\N}$ of
    pieces in $\calA$ such that $A_i\cap A_j=\emptyset$ ($i\neq j$) and
    $\bigcup_{i \in \N } A_i \in \calA$, one has $v(\bigcup_{i \in \N} A_i
    ) = \sum_{i \in \N} v(A_i)$;
  \item[$\mathrm{(D)}$] \textbf{divisible} if for every $A \in \calA$
    and for every real number~$\alpha$, $0 \leq \alpha \leq 1$, there
    exists some $A_\alpha \in \calA$ with $A_\alpha \subseteq A$ such that
    $v(A_\alpha) = \alpha v(A)$.
  \end{enumerate}
\end{definition}

Clearly, $(\Sigma)$ implies $\mathrm{(A)}$~-- take $A_1=A$, $A_2 = B$,
and $A_i=\emptyset$ for $i\geq 3$~-- and $\mathrm{(A)}$ is equivalent
to the so-called \textbf{strong additivity}, defined as $v(A\cup B) =
v(A)+v(B) - v(A\cap B)$: Just observe that $A\cup B =
[A\setminus(A\cap B)] \cup [B\setminus(A\cap B)] \cup [A\cap B]$,
i.e., $A\cap B\neq\emptyset$ counts towards both $v(A)$ and $v(B)$ but
only once in $v(A\cup B)$, hence the correction $-v(A\cap B)$.
Finally, (strong) additivity implies monotonicity.

The assumption that $\calA$ is an algebra makes sure that we can
indeed perform all of the above manipulations with sets without ever
leaving $\calA$. Note, however, that $(\Sigma)$ and $\mathrm{(D)}$
require a certain richness assumption on $\calA$, which need not be
satisfied for an algebra; for example, a union of countably many
member sets need not be in the algebra. In other words: The properties
$(\Sigma)$ and $\mathrm{(D)}$ affect both $v$ and $\calA$.

\begin{remark}\label{rem:content}
  Let $\abscake$ be a(n abstract) cake and $\calA \subseteq
  \mathfrak{P}(\abscake)$ an algebra over~$\abscake$. Any additive
  valuation is a \textbf{finitely additive measure} with total mass
  $v(\abscake)=1$ (see, e.g.,~\citeR[Chapter~4]{schi:b:measures}).
\end{remark}

Requirement $\mathrm{(D)}$ not only demands more from $\calA$ but also
from~$v$. Specifically, $\mathrm{(D)}$ entails that any $N\in\calA$
which does not contain a nonempty and strictly smaller piece of
cake~-- this is an \textbf{atom}, i.e., an indivisible
piece of cake~-- must have zero valuation.

\begin{definition}\label{def-atom}
  Let $\calA$ be an algebra over a(n abstract) cake $\abscake$ and $v$ be a
  finitely additive valuation. A set $A\in\calA$ is an \textbf{atom} if
  $v(A)>0$ and every $B\subseteq A$, $B\in\calA$, satisfies $v(B)=\alpha
  v(A)$ with $\alpha=0$ or $\alpha=1$.
\end{definition}

Clearly, a valuation $v$ which enjoys property $\mathrm{(D)}$ cannot
have atoms.

\subsection{Dividing the Standard Cake}
\label{chap:preliminaries:standard-cake}

Let us briefly discuss the consequences of the notions introduced in
the previous section if $\abscake$ is the standard cake $\concake$.
If, in addition, $\calA$ contains all intervals of type $\langle
a,b\rangle$, then all singletons $\{a\} = [a,b]\setminus (a,b]$ are
in~$\calA$, and they are the only possible atoms. In this case,
$\mathrm{(D)}$ entails that $v$ does not charge single points:
$v(\{a\})=0$ for all $a\in \concake$. This is the proof of the following
lemma.

\begin{lemma}\label{lem:atom-free}
  Let $\concake$ be the standard cake and $\calA$ an algebra of admissible
  sets. Every additive valuation function $v \colon \calA \to [0,1]$
  that satisfies $\mathrm{(D)}$ is atom-free. In particular, if
  $\calA\supset\calI(\concake)$ contains all intervals, then $v(\{a\})=0$
  for all $a\in\concake$.
\end{lemma}

Quite often, we require valuation functions to satisfy
\textbf{continuity}, a property that is crucial for so-called
\emph{moving-knife} cake-cutting protocols to work.
\begin{definition}
  Let $v \colon \calA \to [0,1]$ be a finitely additive valuation
  function on the algebra $\calA = \calI(\concake)$ of finite unions of
  intervals from~$\concake$.
  \begin{enumerate}
  \item The function $x\mapsto F_v(x) := v([0,x])$, $x\in \concake$, is the
    \textbf{distribution function} of the valuation~$v$.
  \item The valuation $v$ is said to be \textbf{continuous} if
    $x\mapsto F_v(x)$ is continuous.
  \end{enumerate}
\end{definition}
Since $v$ is additive, $F_v \colon \concake \to [0,1]$ is positive,
monotonically increasing, and bounded by $F_v(1)=1$.  Note that a
continuous valuation function on $\calI(\concake)$ cannot have atoms, as
\begin{gather*}
  v(\{x\}) = v([0,x]\setminus[0,x))=F_{v}(x)-F_{v}(x-)=0,
  \ \text{where}\
  F_{v}(x-) = \lim_{y\uparrow x}F_{v}(y).
\end{gather*}
The continuity of $v$ can also be cast in the following way: For all
$a$ and $b$ with $0\leq a < b \leq 1$ satisfying $v([0,a]) = \alpha$
and $v([0,b]) = \beta$, and for every $\gamma \in [\alpha, \beta]$,
there exists some $c \in [a,b]$ such that $v([0,c]) = \gamma$.  This
explains the close connection between continuity and divisibility
of~$v$.  In fact, assuming divisibility $\mathrm{(D)}$ of $v$, it can
be shown that the distribution function is necessarily continuous. The
following proof of this statement is inspired by
\citeA[Example~3.4]{sch-sto:t:continuity-assumptions-in-cake-cutting}.

\begin{lemma}\label{contdf}
  Let $v$ be an additive valuation for the standard cake $\concake$, where
  $\calI(\concake)$ denotes the family of admissible pieces.  If $v$ is
  divisible, then the distribution function $F=F_v$ is a continuous
  function with $F(0)=0$.
\end{lemma}
\begin{proof}
  We have seen in Lemma~\ref{lem:atom-free} that a divisible additive
  valuation $v$ has no atoms, so $F(0)=v(\{0\})=0$. Since $F$ is
  monotone and bounded, the one-sided limits $F(t-) := \lim_{s\uparrow
    t}F(s)$ and $F(u+):=\lim_{s\downarrow u}F(s)$ exist for all $t\in
  (0,1]$ and $u\in [0,1)$.

  Assume that $F$ is not continuous. Then there exists some
  $t_{0}\in\concake$ such that $F(t_{0}-)<F(t_{0})$ or
  $F(t_{0}+)>F(t_{0})$.  If $F(t_{0})-F(t_{0}-)=\varepsilon>0$, then
  there exists some $t_{1}<t_{0}$ such that
  $F(t_{0})-F(t_{1})\leq\frac32\varepsilon$. Set $I:=(t_{1},t_{0}]$ and
  observe that $v(I) = F(t_0)-F(t_1) \in
  \left[\varepsilon,\frac32\varepsilon\right]$. Pick an arbitrary
  $J\in\calI(\concake)$ which is contained in $I$. Since $J$ is a finite
  union of intervals, $J$ differs from its closure $\bar{J}$
  by at most finitely many points; as $v(\{x\})=0$ for any
  $x\in\concake$, we have $v(J)=v(\bar{J})$.

  We distinguish
  two cases: If $t_{0}\in\bar{J}$ is not an
  isolated point, then $v(J)=v(\bar{J})\geq\varepsilon$.  If
  $t_{0}\not\in\bar{J}$ or if $t_{0}\in\bar{J}$ is an isolated point,
  then we have due to $v(\{t_{0}\})=0$ that
  \begin{align*}
    v(J)
    &= v(\bar{J}) \leq F(t_{0}-)-F(t_{1})
     = \big(F(t_{0})-F(t_{1})\big) - \big(F(t_{0})-F(t_{0}-)\big)\\
    &= v(I)-\varepsilon \leq \tfrac12\varepsilon.
  \end{align*}
  Hence, it is not possible to select a piece of cake
  $J\in\calI(\concake)$ with $J\subseteq I$ and $v(J)=\frac34\varepsilon\in
  \left[\frac12\cdot v(I),\frac34\cdot v(I)\right]$, which contradicts
  divisibility.

  If $F(t_{0}+)-F(t_{0})=\varepsilon>0$, a similar argument applies.
\end{proof}

Conversely, if the distribution function $F_v$ of a finitely additive
valuation $v$ defined on $\calI(X)$ is continuous with $F_v(0)=0$,
then it is easy to see that $v$ is divisible. Hence we get:

\begin{corollary}\label{cor-D-vs-sigma}
  A finitely additive valuation $v$ on $\calI(\concake)$ is divisible if,
  and only if, its distribution function $F_v$ is continuous with
  $F_v(0)=0$.  This is also equivalent to $v$ being atom-free.
\end{corollary}

Corollary~\ref{cor-D-vs-sigma} establishes a one-to-one correspondence
between divisible valuations and monotonically increasing, continuous
functions on $\concake$ which are $0$ at the origin and $1$ at
$x=1$. This shows that the identity $x\mapsto x$ gives rise to a
valuation (it assigns every interval $\langle a,b\rangle$ its natural
length $b-a$) but also the Cantor function $V(x)$ from
Example~\ref{ex:cantor-function} can be viewed as a valuation
function.

We will see in the next section that every finitely additive,
divisible valuation can be extended to become and identified with a
unique $\sigma$-additive measure that is defined on the Borel
$\sigma$-algebra $\calB(\abscake)$; this is the smallest family of sets
that contains all intervals and that is stable under complements and
countable unions of its members. This enables us to evaluate sets in
$\calB(\abscake)$ that are not finite unions of intervals, such as the
Cantor set in Example~\ref{ex:cantor-set}.

\subsection{Measure Theory: The Art of Dividing a Cake by Countably
  Many Cuts}\label{chap:preliminaries:sigma}

In Sections~\ref{chap:preliminaries:finite}
and~\ref{chap:preliminaries:standard-cake}, we have focused on
finitely many cuts when dividing the cake. But we may easily come into
the situation where the number of cuts is not limited; not all
protocols in the cake-cutting literature are finite.\footnote{ For
  example, prior to the celebrated finite bounded envy-free cake-cutting
  protocol due to
  \citeA{azi-mac:c:discrete-bounded-envy-free-cake-cutting-protocol,azi-mac:j:bounded-envy-free-cake-cutting-algorithm},
  the cake-cutting protocol of \citeA{bra-tay:j:protocol} was the best
  protocol known to guarantee envy-freeness for any number of players.
  They argue that the allocation must become envy-free at some (unknown)
  finite stage, which is why their protocol is considered to be a
  \emph{finite unbounded} envy-free procedure only.  And yet, being
  open-ended, it is in some sense even an \emph{infinite} procedure that
  describes an infinite process.  Similarly, it is reasonable to
  conjecture that some moving-knife procedures can be converted to
  discrete procedures that require infinitely many cuts.
  \label{foo:infinite-protocols}} Thus we are led to consider unions
of countably many pieces and the valuation of such countable unions,
see also property $(\Sigma)$ in
Definition~\ref{def:valuation-function}.  To deal with such
situations, measure theory provides the right tools.

We will now introduce some basics from measure theory, which we need
in the subsequent discussion of the cake-cutting literature.  Our
standard references for measure theory are the monographs by
\citeA{schi:b:measures} and \citeA{schi:b:counterexamples}, where also
further background information can be found.

\begin{definition}\label{def:sigma-algebra}
  Let $\concake$ be a cake.  A subset $\calA\subseteq
  \mathfrak{P}(\concake)$ is called a
  \textbf{$\boldsymbol{\sigma}$-algebra} over $\concake$ if $\calA$ is an
  algebra over $\concake$ and, for all sequences $(A_n)_{n \in \N}$ with
  $A_n \in \calA$, the countable union $\bigcup_{n \in \N} A_n$ is in
  $\calA$, too.
\end{definition}

Every algebra in
$\concake$ containing finitely many sets is automatically a
$\sigma$-algebra. On the other hand, $\mathfrak{P}(\concake)$ is both an
algebra and a $\sigma$-algebra, whereas the family $\calI(\concake)$ is an
algebra, but not a $\sigma$-algebra: For instance, the Cantor dust
$C_p$ (cf.\ Example~\ref{ex:cantor-set}) is not in
$\calI(\concake)$. Recall that we defined $\calI(\concake)$ to be the
smallest algebra containing all (finite unions of) intervals in
$\concake$; thus it is natural to consider the smallest $\sigma$-algebra
containing all (finite unions of) intervals in $\concake$.

To see that this is well-defined, we need a bit more notation. Recall
that $\langle a,b\rangle$ stands for any (open, closed, or half-open)
interval of~$\concake$.  We denote by
\begin{gather*}
  \calQ(\concake) = \left\{\langle a,b\rangle \mid a,b\in \concake\right\}.
\end{gather*}
the \emph{family of all intervals within~$\concake$}.

Moreover, if $\calP\subseteq\mathfrak{P}(\concake)$ is any family, then
$\sigma(\calP)$ denotes the smallest $\sigma$-algebra containing
$\calP$. This can be a fairly complicated object and its existence is
not really obvious. To get an idea as to why $\sigma(\calP)$ makes
sense, we note that $\calP\subseteq\mathfrak{P}(\concake)$, that
$\mathfrak{P}(\concake)$ is a $\sigma$-algebra, and that the intersection
of any number of $\sigma$-algebras is still a $\sigma$-algebra.

The next lemma is a standard result from measure theory.
\begin{lemma}\label{lem:Borel-sigma-algebra}
  Let $\calP$ denote any of the four families of open intervals,
  closed intervals, left-open intervals, or right-open intervals within
  $\concake$. It holds that
  \begin{align*}
    \sigma(\calP) =  \sigma(\calQ(\concake)).
  \end{align*}
\end{lemma}

The fact that $\sigma(\calQ(\concake))$ coincides with the
$\sigma$-algebra generated by all closed intervals in $\concake$ can
be used to generalize Lemma~\ref{lem:Borel-sigma-algebra} to abstract
cakes, which carry a topology, hence a family of open and of closed
sets. The thus generated \enquote{topological} $\sigma$-algebra plays
a special role and has a special name.

\begin{definition}\label{def:Borel-sigma-algebra}
  We denote by $\calB(\concake)$ the smallest $\sigma$-algebra on $\concake$
  containing all closed intervals from $\concake$ and call it the
  \textbf{Borel} or \textbf{topological $\boldsymbol{\sigma}$-algebra}
  over~$\concake$.
\end{definition}

The following definition is also well-known.  We state it only for the
standard cake, but it is clear how to extend it to
abstract cakes.

\begin{definition}\label{def:measure}
  Let $\concake$ be a cake and $\calA$ a $\sigma$-algebra on $\concake$. A
  (positive) \textbf{measure} $\mu$ on $\concake$ is a map $\mu \colon
  \calA \to [0,\infty]$ satisfying that $\mu(\emptyset) = 0$ and $\mu$
  is $\sigma$-additive.
\end{definition}

It is useful to see a measure $\mu$ as a function defined on the
sets. If the set-function is additive, then $\sigma$-additivity is, in
fact, a continuity requirement on $\mu$, as it allows to interchange
the limiting process in the infinite union $\bigcup_{i\in\N} A_i$ of
pairwise disjoint sets with a limiting process in the sum. To wit:
\begin{gather*}
    \mu \left( \bigcup_{n \in \N} A_n \right)
    = \lim_{N\to\infty} \mu \left( \bigcup_{n=1}^N A_n \right)
    = \lim_{N\to\infty}  \sum_{n=1}^N\mu \left( A_n \right)
    = \sum_{n\in\N} \mu \left( A_n \right);
\end{gather*}
since all terms are positive,  the value of the sum is well-defined, i.e., it is
either convergent in $[0,\infty)$ or improperly convergent yielding $+\infty$.
Equivalently, we can state $\sigma$-additivity as $B_1\subset
B_2\subset B_3\subset \cdots \uparrow B = \bigcup_{n\in\N} B_n$, then
$\mu(B_n)\uparrow \mu(B)$ (for any measure $\mu$) or as $C_1\supset
C_2\supset C_3\supset\cdots \downarrow C = \bigcap_{n\in\N} C_n$, then
$\mu(C_n)\downarrow \mu(C)$ (for finite measures $\mu$).

Sometimes (and a bit provocatively) it is claimed that there are
essentially only two measures on $\concake$ (or on $\RR$ or $\RR^n$):
\textbf{Lebesgue measure} $A\mapsto \lambda(A)$ and \textbf{Dirac
measure} $A\mapsto \delta_x(A)$, where $x\in\concake$ is a fixed
point. Let us briefly discuss these two extremes and explain as to why
the claim is incorrect but still sensible.

\paragraph*{Dirac Measure.}
Let $a\in\concake$ be a fixed point and set $A\mapsto \delta_a(A) = 1$ or
$=0$ according to $a\in A$ or $a\notin A$, respectively. This
definition works for any $A\subseteq\concake$, and it is easy to see that
this set-function is indeed a measure (in the sense of
Definition~\ref{def:measure} on the $\sigma$-algebra
$\calA=\mathfrak{P}(\concake)$~-- or any smaller $\sigma$-algebra over
$\concake$.

We call $\{a\}$ the \textbf{support} of $\delta_a$ since, by
definition, $\delta_a$ charges only sets such that $\{a\}\subseteq
A$. If we compare Dirac measure with Lebesgue's measure, the problem
is that the support of $\delta_a$ is a degenerate interval $\{a\} =
[a,a]$ of length zero, see below.

\paragraph*{Lebesgue Measure.}
The idea behind Lebesgue measure is to have a set-function
$A\mapsto\lambda(A)$ in $\concake$ (or in $\RR$ or $\RR^n$) with all
properties of the familiar volume from geometry; in particular, we
want a volume that is additive and invariant under shifts and
rotations. Thus it is natural to define for a simple set $Q$ like an
interval $Q=(a,b]\subset\concake$ (or an $n$-dimensional \enquote{cube}
$Q =\mathop{\times}\limits_{i=1}^n (a_i,b_i]$)
\begin{gather*}
  \lambda(Q) = b-a \\
  \quad\left(\text{respectively,}\quad
    \lambda(Q) = \prod_{i=1}^n (b_i-a_i) =
    \text{length}\times\text{width}\times\text{height}\times\cdots
  \right).
\end{gather*}
Invariance under shifts together with the $\sigma$-additivity
$(\Sigma)$ allow us to exhaust (\enquote{triangulate}) more
complicated shapes like a circle with countably many disjoint sets
$(Q_n)_{n\in\N}$ such that with $A=\bigcup_{n\in\N}Q_n$, we have
$\lambda(A)=\sum_{n\in\N}\lambda(Q_n)$.  The restriction to countable
unions is natural, as we exhaust a given shape by nontrivial sets
$Q_n$, having nonempty interior: Each of them contains a rational
point $q\in\QQ^n$; hence, there are at most countably many
nonoverlapping~$Q_n$.

There are immediate questions with this approach: Which types of sets
can be \enquote{measured}? Is the procedure unique? Is the process of
measuring more complicated sets constructive?  At this point we
encounter a problem: General sets $A\subseteq\RR^n$ are way too
complicated to get a well-defined and unique extension of $\lambda$
from the rectangles to $\mathfrak{P}(\RR^n)$.  In dimension $n=1$ and
for the standard cake $\concake$, the Cantor sets $C_p$ from
Example~\ref{ex:cantor-set} were already challenging, but the Vitali
set from Example~\ref{ex-vitali} shows that the cocktail of shift
invariance and $\sigma$-additivity becomes toxic.

The way out is the notion of measurable sets and Carath\'eodory's
extension theorem (stated as Theorem~\ref{thm:cara} further down).
This works as follows: In view of the $\sigma$-additivity property
of~$\lambda$, it makes sense to consider the $\sigma$-algebra
$\calA\subseteq\mathfrak{P}(\RR^n)$ which contains the intervals
(respectively, cubes).  Thus we naturally arrive at the notion of the
Borel $\sigma$-algebra as the canonical domain of Lebesgue
measure. Unfortunately, there are so many Borel sets that we cannot
build them constructively from rectangles~-- we would need transfinite
induction for this~-- and this is one of the reasons why cutting a
cake is not always a piece of cake.

The question of whether \emph{every} set $A\subseteq\RR^n$ has a
unique geometric volume (in the above sense) is
dimension-dependent. If $n=1$ or $n=2$, we can extend the notion of
length and area to all sets, but not in a unique way. In dimension $3$
and higher, we'll end up with contradictory statements (such as the
\emph{Banach--Tarski paradox}; see,
e.g.,~\citeR{wag:b:banach-tarski-paradox}) if we try to have a
finitely additive geometric volume for all sets. This conundrum can be
resolved by looking at the Borel sets or the Lebesgue sets~-- these
are the Borel sets enriched by all subsets of Borel sets with Lebesgue
measure zero.

\paragraph*{General Measures.}
Let us return to the assertion that $\lambda$ and $\delta_a$ are
\enquote{essentially the only measures} on~$\RR^n$. To keep things
simple, we discuss here only the standard cake $\concake$.

Lebesgue's decomposition theorem shows that all $\sigma$-additive
measures $\mu$ on $\concake$ with the Borel $\sigma$-algebra
$\calB(\concake)$ are of the form $\mu = \mu^{\mathrm{ac}} +
\mu^{\mathrm{sc}} + \mu^{\mathrm{d}}$ where \enquote{ac,}
\enquote{sc,} and \enquote{d} stand for absolutely continuous,
singular continuous, and discontinuous. This is best explained by
looking at the distribution function $F(x)=F_\mu(x)$. Since $x\mapsto
F(x)$ is increasing, it is either continuous or discontinuous (with at
most countably many discontinuities), accounting for the parts
($\mu^{\mathrm{ac}}, \mu^{\mathrm{sc}}$) and $\mu^{\mathrm{d}}$,
respectively. At the points where $F$ is continuous, we have again two
possibilities: $F$ is either differentiable ($F'(x) = f(x)$) or it
isn't, yielding \enquote{ac} vs.\ \enquote{sc.} From Lebesgue's
differentiation theorem it is known that the points with \enquote{sc}
or \enquote{d} must have Lebesgue measure zero. Thus, we finally
arrive at the decomposition
\begin{gather}\label{ac-sc-d}
    \mu(dx) = f(x)\,dx + \mu^{sc}(dx) + \sum_{i}
    (F(x_i)-F(x_i-))\,\delta_{x_i}(dx),
\end{gather}
where $x_1,x_2,\dots$ are the at most countably many discontinuities
(jump points) of $F$ and $f(x) = \frac{d}{dx}F(x)$.

Here are four typical \textbf{examples} for valuations corresponding
to these cases.
\begin{itemize}
\item Purely \textbf{ac}: $F^{\mathrm{ac}}(x) := x$ is absolutely
  continuous since $\frac{d}{dx}F^{\mathrm{ac}}(x) = 1$ exists and
  $F^{\mathrm{ac}}(x) = \int_0^x 1\,dt$. This $F^{\mathrm{ac}}$
  corresponds to Lebesgue measure. In general, an absolutely
  continuous $F^{\mathrm{ac}}(x)$ is always of the form
  $F^{\mathrm{ac}}(x) - F^{\mathrm{ac}}(0) = \int_0^x f(t)\,dt$ and
  $f(t) = \frac{d}{dt}F^{\mathrm{ac}}(t)$.
\item Purely \textbf{sc}: The Cantor function
  $F^{\mathrm{sc}}(x):=V(x)$ from Example~\ref{ex:cantor-function} is
  continuous, but it is not absolutely continuous. $V'(x)$ exists (in a
  classical sense) only in the points $\concake\setminus C_{\nicefrac
    13}$ and $V(x) \neq \int_0^x V'(t)\,dt$. The corresponding valuation
  is nevertheless of the form $v((a,b]) = V(b)-V(a)$, but it cannot be
  represented in the form $\int_a^b f(t)\,dt$ for any function~$f$.
\item Purely \textbf{d}: Any increasing step-function with jumps of
  size $\Delta_i > 0$, $i=1,2,\dots$, at the points $x_i\in\concake$
  corresponds to the discontinuous case: We have atoms exactly at the
  points $x_i$ where $F^{\mathrm{d}}(x)$ is discontinuous (i.e.,
  jumps). The general form of such functions is $F^{\mathrm{d}}(x) =
  \sum_{i=1}^\infty \Delta_i \mathbf{1}_{[x_i,1]}(x)$ where
  $\mathbf{1}_{[x_i,1]}(x)$ is the indicator function (taking the values
  $1$ and $0$ according to $x\in[x_i,1]$ or $x\notin [x_i,1]$,
  respectively) and $\sum_{i=1}^\infty \Delta_i = 1$.
\item Mixed \textbf{ac+cs+d}: Let $p_{\mathrm{ac}}, p_{\mathrm{sc}},
  p_{\mathrm{d}}\in [0,1]$ be such that $p_{\mathrm{ac}}+
  p_{\mathrm{sc}}+ p_{\mathrm{d}} = 1$ and let $F^{\mathrm{ac}},
  F^{\mathrm{sc}}, F^{\mathrm{d}}$ be as in the previous examples.
  Then the convex combination $F(x) = p_{\mathrm{ac}}
  F^{\mathrm{ac}}(x) + p_{\mathrm{sc}}F^{\mathrm{sc}}(x)+
  p_{\mathrm{d}}F^{\mathrm{d}}(x)$ corresponds to a valuation which
  combines all three types of (dis-)continuity properties.
\end{itemize}

Let us close this section with the central result on the extension of
valuations defined on an algebra $\calA$ to measures on the
$\sigma$-algebra $\sigma(\calA)$ generated by~$\calA$.  We state it
only for the standard cake; the formulation for more abstract cakes is
obvious.

\begin{theorem}[Carath\'eodory's extension theorem]\label{thm:cara}
 Let $\calA$ be the algebra of admissible pieces of the cake $\concake$ and
$v:\calA\to[0,1]$ be a valuation  such that $v(\emptyset)=0$.
If $v$ is additive and $\sigma$-additive
  relative to $\calA$, i.e., $v$ satisfies ${(\Sigma)}$, then there is a
  unique extension of $v$, defined on $\sigma(\calA)$, which is a
  $\sigma$-additive measure on $\sigma(\calA)$.
\end{theorem}

\subsection{Abstract Cakes}
\label{subsec:abstract}

Let us briefly discuss more general cakes $\abscake$ than
$\concake$. In this section, $\abscake\neq\emptyset$ will be a general
set, $\calA$ an algebra of admissible pieces.  The notion of
$\sigma$-algebra is, \emph{mutatis mutandis}, the same as in the case
of the standard cake (Definition~\ref{def:sigma-algebra}) and we
denote by $\sigma(\calA)$ the smallest $\sigma$-algebra that contains
the algebra~$\calA$.  The definition and the properties of a valuation
$v\colon \calA \to [0,1]$ (cf.\
Definition~\ref{def:valuation-function}) still work in this general
setting, but since $\abscake$ is abstract, there may not be (an
equivalent of) a distribution function; this means that the connection
between divisibility and $\sigma$-additivity, cf.\ Lemma~\ref{contdf}
and Corollary~\ref{cor-D-vs-sigma}, might fail in an abstract setting.

We begin with a new definition of $\mathrm{(D)}$ for finitely additive
valuations on abstract cakes.

\begin{definition}\label{def-DD}
  A finitely additive valuation $v$ on an abstract cake $\abscake$ and an
  algebra of admissible pieces $\calA$ has the property $\mathrm{(DD)}$
  if for every $A\in\calA$ and $\alpha\in(0,1)$, there is an increasing
  sequence of sets $B^1_\alpha \subset B^2_\alpha \subset B^3_\alpha
  \subset \cdots$, $B^n_\alpha \in\calA$, such that $B^n_\alpha\subset
  A$ and $\sup_{n\in\N} v(B^n_\alpha) = \alpha v(A)$.
\end{definition}

Property $\mathrm{(DD)}$ essentially says that for every value $\alpha
v(A)\in [0,1]$ we can find an admissible piece of cake
$B^n_\alpha\subset A$ whose valuation $v(B^n_\alpha)$ is close to
$\alpha v(A)$.  The limiting piece $\bigcup_n B^n_\alpha$, which
should produce the value $\alpha v(A)$ exactly, may not be admissible
if we are restricted to finitely many cuts.

If $v$ is a $\sigma$-additive valuation and $\calA$ a
$\sigma$-algebra, then $B_\alpha := \bigcup_{n\in\N} B_\alpha^n$ is
again in $\calA$, and, because of $\sigma$-additivity, we see that
$v(B_\alpha) = \sup_{n\in\N} v(B_\alpha^n)$. Thus the properties
$\mathrm{(D)}$ and $\mathrm{(DD)}$ are indeed equivalent for
$\sigma$-additive valuations (or, in view of
Corollary~\ref{cor-D-vs-sigma}, for finitely additive valuations on
the standard cake $\concake$ and $\calA \supset \calI(\concake)$).

We will also need the opposite of the property~$\mathrm{(DD)}$; to
this end, recall Definition~\ref{def-atom} of an atom.  If $A$ and $B$
are atoms, then we have either $v(A\cap B)=0$ or $v(A\cap
B)=v(A)=v(B)>0$; in the latter case, if $v(A\cap B)>0$, we call the
atoms \emph{equivalent}. If $A$ and $B$ are nonequivalent, then $A$
and $B\setminus A$ are still nonequivalent and disjoint. Iterating
this procedure, we can always assume that countably many nonequivalent
atoms $(A_n)_{n\in\N}$ are disjoint: Just replace the atoms by $A_1,
A_2\setminus A_1,\dots, A_{n+1}\setminus\bigcup_{i=1}^n A_i, \dots$.

Since $v(\abscake)=1$, a finitely additive valuation $v$ can have at most $n$
nonequivalent atoms such that $v(A)\geq \frac 1n$, and so there are at
most countably many atoms. Comparing Definition~\ref{def-DD} which
defines property~$\mathrm{(DD)}$ with Definition~\ref{def-atom} of an
atom, it is clear that $\mathrm{(DD)}$ implies that $v$ has no atoms.
We will see in Theorem~\ref{theo-conti} that the converse implication
holds as well.

\begin{definition}\label{def-slice}
  Let $v$ be a finitely additive valuation on the algebra $\calA$ over
  a(n abstract) cake~$\abscake$. The valuation $v$ is \textbf{sliceable}
  if for any $\varepsilon>0$, there are finitely many disjoint sets
  $B_i\in\calA$, $i=1,\dots, n$, $n=n(\varepsilon)$, such that $0<
  v(B_i)\leq\varepsilon$ and $X=B_1\cup\dots\cup B_n$.

  A set $B\in\calA$ is \textbf{$\boldsymbol{v}$-sliceable} if the
  set-function $A\mapsto v(A\cap B)$ is sliceable.
\end{definition}

We will now see that a sliceable finitely additive valuation
enjoys property~$\mathrm{(DD)}$, and \emph{vice versa}, i.e.,
sliceability, atom-freeness, and property~$\mathrm{(DD)}$ are pairwise
equivalent for finitely additive valuations.

\begin{theorem}\label{theo-conti}
  Let $v$ be a finitely additive valuation on an algebra $\calA$ over
  a(n abstract) cake~$\abscake$. The conditions $\mathrm{(DD)}$, \enquote{$v$
    is sliceable,} and \enquote{$v$ has no atoms} are pairwise equivalent.
\end{theorem}

\begin{proof}
  We start by showing that atom-freeness implies sliceability. Fix
  $\varepsilon>0$.

  \smallskip\textbf{Step~1:}
  Let $Y\subseteq \abscake$ be \emph{any} subset, and assume that there
  is some $B\subseteq Y$, $B\in\calA$, such that $v(B)>0$. Define
  \begin{gather*}
    \calF^Y :=\calF^Y_\varepsilon := \{F\in\calA \mid F\subseteq Y,
    \: 0< v(F)\leq\varepsilon\}.
  \end{gather*}
  We claim that for the special choice $Y=B\in\calA$ the family
  $\calF^B$ is not empty.

  Since $B$ is not an atom, there is some $F\subseteq B$, $F\in\calA$,
  with $0< v(F)< v(B)$.

  If $v(F)\leq\varepsilon$, then $F\in\calF^B$, and we are done.

  If $v(F)>\varepsilon$, we assume, to the contrary that there is no
  subset $F'\subseteq F$, $F'\in\calA$, with $0<
  v(F')\leq\varepsilon$. Since $F$ cannot be an atom, there is a subset
  $F'\subseteq F$ with $\varepsilon< v(F')< v(F)$ and $v(F\setminus
  F')>\varepsilon$. Iterating this with $F\rightsquigarrow F\setminus
  F'$ furnishes a sequence of disjoint sets $F_1=F',F_2,F_3,\dots$ with
  $v(F_i)>\varepsilon$ for all $i\in\N$. This is impossible since
  $v(F)<\infty$. So we can find some $F'\subseteq F\subseteq B$ with $0<
  v(F')\leq\varepsilon$, i.e., $\calF^B$ is not empty.

  \medskip
  \textbf{Step~2:}
  Define a(n obviously monotone) set-function $c(Y) :=
  \sup_{C\in\calF^Y} v(C)$ for any $Y\subseteq \abscake$; as usual,
  $\sup\emptyset = 0$. Since $\calF^\abscake$ is not empty, we can pick
  some $B_1\in\calF^\abscake$ such that $\frac 12 c(\abscake) < v(B_1)
  \leq\varepsilon$.

  If $v(\abscake\setminus B_1)\leq\varepsilon$, we set $B_2 :=
  \abscake\setminus B_1$; otherwise, we can pick some
  $B_2\in\calF^{\abscake\setminus B_1}$ such that $\frac 12
  c(\abscake\setminus B_1) < v(B_2) \leq\varepsilon$.

  In general, if $v(\abscake\setminus (B_1\cup\dots\cup
  B_n))\leq\varepsilon$, we set $B_{n+1}=\abscake\setminus
  (B_1\cup\dots\cup B_n)$; otherwise, we pick
  \begin{equation}\label{e-cond}
    B_{n+1}\in\calF^{\abscake\setminus (B_1\cup\dots\cup B_n)}
    \quad\text{such that}\quad
    \frac 12 c(\abscake\setminus (B_1\cup \dots \cup B_n)) \leq  v(B_{n+1})\leq \varepsilon.
  \end{equation}

  We are done if this procedure stops after finitely many steps;
  otherwise, we get a sequence of disjoint sets $B_1, B_2, \dots$
  satisfying~\eqref{e-cond}.
  Define $B_\infty := \abscake\setminus\bigcup_n B_n$. This set need
  not be in~$\calA$, but we still have, because of \eqref{e-cond},
  \begin{gather*}
    c(B_\infty)
    \leq c(\abscake\setminus (B_1\cup\dots\cup B_m))
    \leq 2 v(B_{n+1})
    \xrightarrow[n\to\infty]{}0
  \end{gather*}
  since the series
  \begin{gather*}
    \sum_{n\in\N}  v(B_n)
    = \sup_N\sum_{n=1}^N  v(B_n)
    = \sup_N v\left(\bigcup_{n=1}^N B_n\right)
    \leq  v(\abscake)
  \end{gather*}
  converges.
  In particular, $\lim_{n\to\infty} v(\abscake\setminus \bigcup_{i=1}^nB_i)=0$.

  Using again the convergence of the series $\sum_n v(B_n)$, we find
  some $N=N(\varepsilon)$ such that $\sum_{n > N}
  v(B_n)\leq\varepsilon$, hence $B_1,B_2,\dots, B_N$ and
  $\abscake\setminus\bigcup_{n=1}^N B_n$ are the desired small pieces
  of~$\abscake$. This completes the proof that $v$ is sliceable.

  \medskip
  We now show that sliceability implies condition~$\mathrm{(DD)}$.
  Let $B\in\calA$ with $v(B)>0$.
  Since the \enquote{relative} finitely additive valuation
  $v_B(A):= v(A\cap B)/ v(B)$ inherits the nonatomic property from $v$,
  it is clearly enough to show that for every $\alpha\in(0,1)$, there is
  an increasing sequence
  \begin{gather*}
    B^1_\alpha\subset B^{2}_\alpha\subset B^{3}_\alpha\subset\cdots,
    \quad B^n_\alpha\in\calA \::\: \sup_{n\in\N} v(B^n_\alpha) = \alpha,
  \end{gather*}
  which is the property $\mathrm{(DD)}$ relative to the full
  cake $\abscake$ only.

  Since $v$ is sliceable,
  there are mutually disjoint sets
  $C_1^{n}, \dots, C_N^{n}\in\calA$, where $N=N(n)$,
  $\abscake=\bigcup_{i=1}^N C_i^{n}$, and $v(C_i^{n})<\frac 1n$.

  Let $k = \lfloor {1}/{\alpha} \rfloor +1$.  Set $B_k :=
  C_1^{k}\cup\dots\cup C_{M(k)}^{k}$, where $M(k)\in \{1,\dots, N(k)\}$
  is the unique number such that
  \begin{gather*}
    \sum_{i=1}^{M(k)} v(C_i^k)\leq\alpha <\sum_{i=1}^{M(k)+1}
    v(C_i^k)\leq \sum_{i=1}^{M(k)} v(C_i^k)+\frac 1k.
  \end{gather*}

  By construction, $\alpha\geq v(B_k)=\sum_{i=1}^{M(k)} v(C_i^k) >
  \alpha-\frac 1k$.  Thus, we can iterate this procedure, considering
  $\abscake\setminus B_k$ and constructing a set $D_{k+1}\subseteq
  \abscake\setminus B_k$ that satisfies
  \begin{gather*}
    (\alpha- v(B_k)) \geq v(D_{k+1}) > (\alpha- v(B_k))-\frac 1{k+1}.
  \end{gather*}
  For $B_{k+1}:= B_k\cup D_{k+1}$, we get $\alpha\geq
  v(B_{k+1})>\alpha-\frac 1{k+1}$.

  The sequence $B_{k+i}$, $i\in\N$, satisfies
  $v(B_{k+i})\uparrow\alpha$, i.e., $B_\alpha^n = B_{k+n}$ is the
  sequence of sets we need to have property~$\mathrm{(DD)}$.

  As mentioned earlier, $\mathrm{(DD)}$ implies atom-freeness, which
  completes this proof.
\end{proof}

Since for a $\sigma$-additive valuation on a $\sigma$-algebra~$\calA$,
properties $\mathrm{(D)}$ and $\mathrm{(DD)}$ are equivalent, we
immediately get:
\begin{corollary}\label{cor-conti}
  Let $v$ be a $\sigma$-additive valuation on a $\sigma$-algebra
  $\calA$ over an abstract cake~$\abscake$. The conditions $\mathrm{(D)}$,
  $\mathrm{(DD)}$, \enquote{$v$ is sliceable,} and \enquote{$v$ has no
    atoms} are pairwise equivalent.
\end{corollary}

If $A_1, A_2, \dots$ is an enumeration of the nonequivalent atoms of
the $\sigma$-additive valuation~$v$, then $A_\infty := \abscake\setminus
\bigcup_{n\in\N}A_n\in\calA$, and we can restate
Corollary~\ref{cor-conti} in the form of a decomposition theorem.
\begin{corollary}\label{cor-deco}
  Let $v$ be a $\sigma$-additive valuation on a $\sigma$-algebra
  $\calA$ over a(n abstract) cake~$\abscake$.  Then $\abscake$ can be written as
  a disjoint union of a $v$-sliceable set $A_{\infty}$ and at most
  countably many atoms $A_1, A_2, \dots$.
\end{corollary}

\section{Which Pieces Should Be Admissible?}
\label{chap:literature}

In the cake-cutting literature, a great variety of different
definitions have been used for the set $\calP$ of admissible pieces of
cake.  We first collect the most commonly used definitions
for~$\calP$, along with the corresponding references and discuss them
in detail.  Then we show several relations among these definitions and
discuss what this implies for a most reasonable choice of~$\calP$.

Typical choices for the set $\calP$ containing all admissible pieces of a
standard cake $\concake$ are
\begin{enumerate}
\item\label{deftype:finite_union_intervals} all finite unions of
  intervals from $\concake$, i.e., the family $\calI(\concake)$ defined
  earlier on page~\pageref{pg:calI-cake};
\item\label{deftype:countable_union_intervals} all countable unions of
  intervals from $\concake$, i.e., $\calI(\concake)^\N = \left\{ \bigcup_{i \in
      \N} I_i \mid I_i \in \calI(\concake) \right\}$;
\item\label{deftype:borel} the Borel $\sigma$-algebra over $\concake$,
  i.e., $\calB(\concake)$;
\item\label{deftype:lebesgue} the set of all Lebesgue sets
  over $\concake$, i.e., $\calL(\concake)$;\footnote{Recall that a set
    $\tilde B$ is a Lebesgue set if, and only if, there is a
    Borel set $B$ such that the symmetric difference
    $\tilde B\vartriangle B := ( B\setminus B ) \cup ( B\setminus
    \tilde B )\subseteq N$ is contained in a Borel set $N$
    with Lebesgue measure $\lambda(N) = 0$. We will see in
    Theorem~\ref{thm:strict-inclusions} that there are indeed
    Lebesgue sets that are not Borel sets.}
  or
\item\label{deftype:powerset} the power set $\mathfrak{P}(\concake)$
  of~$\concake$.
\end{enumerate}

Assuming $\calP = \calI(\concake)$ is common among papers that consider
only finite cake-cutting protocols. Such protocols can make only a
finite number of cuts, thus producing a finite set of contiguous
pieces, i.e., intervals, to be evaluated by the players.  Authors that
make this assumption and use $\calP = \calI(\concake)$ include
\citeA{woe-sga:j:complexity-cake-cutting},
\citeA{str:j:finite-protocols-cannot-ef},
\citeA{lin-rot:c:dgef},
\citeA{pro:c:though-shalt-covet},
\citeA{wal:c:online-cake-cutting},
\citeA{coh-lai-par-pro:c:optimal-envy-free-cake-cutting},
\citeA{bei-che-hua-tao-yan:c:optimal-connected-cake-cutting},
\citeA{cec-pil:j:computability-equitable-divisions},
\citeA{bra-fel-lai-mor-pro:c:maxsum-fair-cake-divisions},
\citeA{cec-dob-pil:j:existence-equitable-divisions},
\citeA{che-lai-par-pro:j:truth-justice-cake},
\citeA{bra-mil:c:equilibrium-analysis-cake-cutting},
\citeA{azi-mac:c:discrete-bounded-ef-protocol-four-players,azi-mac:c:discrete-bounded-envy-free-cake-cutting-protocol,azi-mac:j:bounded-envy-free-cake-cutting-algorithm},
\citeA{edm-pru:c:no-piece-of-cake}, and
\citeA{azi-mac:c:discrete-bounded-envy-free-cake-cutting-protocol}.

As a special case, valuation functions may even be restricted to
single intervals, which is done by
\citeA{cec-pil:j:near-equitable-2-person-cake-cutting-algorithm}
and
\citeA{aum-dom:c:efficiency-connected-pieces}.
Even though the restriction to finite unions of intervals is sensible
from a practical perspective, it may artificially constrain results
that could hold also in a more general setting.

\citeA{bra-pro-zha:c:externalities-cake-cutting}
extend $\calP$ to contain countably infinite unions of intervals, i.e.,
$\calI(\concake)^\N$.

Authors assuming $\calP = \calB(\concake)$ include
\citeA{str-woo:j:measures-agree},
\citeA{den-qi-sab:t:complexity-envy-free}, and
\citeA{seg-nit-has-aum:j:two-dimensional-cake-cutting}.

Works using $\calP = \calL(\concake)$ include those by
\citeA{rei-pot:j:finding-ef-pareto-optimal-division},
\citeA{arz-aum-dom:c:throw-cake}, and
\citeA{rob-web:j:near-exact-and-envy-free-cake-division}.
Additionally, several authors do not explicitly make the assumption
$\calP = \calL(\concake)$, but they define valuation functions based on
(Lebesgue-)measurable sets only, most prominently, a valuation
function is often defined as the integral of a given probability
density function on~$\concake$. This or a similar assumption is made by
\citeA{bra-jon-kla:perfect-cake-cutting-with-money,bra-jon-kla:j:better-cut-cake,bra-jon-kla:j:pie-cutting,bra-jon-kla:j:there-may-be-no-perfect-division},
\citeA{rob-web:b:cake-cutting-algorithms-be-fair-if-you-can},
\citeA{web:j:minimal-cuts},
\citeA{aum-dom-has:c:socially-efficient-cake-divisions},
\citeA{bra-car-kur-pro:t:strategic-fair-division}, and
\citeA{car-lai-pro:c:more-expressive-cake-cutting}.

Papers that assume $\calP = \mathfrak{P}(\concake)$ include those by
\citeA{mac-mar:j:cut-pizza-fairly},
\citeA{sga-woe:j:linear-approximation-cake-cutting},
\citeA{sab-wan:c:envy-free-cake-cutting-for-five},
\citeA{man-oka:c:meta-envy-free-cake-cutting-protocols}, and
\citeA{aum-dom-has:c:truthful-cake-auctions}.

Finally, several works, including those by
\citeA{dub-spa:j:cut-cake-fairly},
\citeA{bar:j:game-theoretic-algorithms-for-cake-divisions,bar:j:super-envy-free-divisions},
\citeA{zen:c:approximate-envy-free-procedures}, and
\citeA{bra-tay:j:envy-free-protocol},
define the set of admissible pieces of cake to be
some ($\sigma$-)algebra (not necessarily Borel) over~$\concake$.

Note that each of the sets $\calI(\concake)$, $\calB(\concake)$,
$\calL(\concake)$, and $\mathfrak{P}(\concake)$ is an algebra over $\concake$, and
all, except $\calI(\concake)$, are also $\sigma$-algebras over~$\concake$.
However, $\calI(\concake)^\N$ is
not an algebra, as the proof of the following theorem shows.

Having introduced all the different approaches currently used in the
literature, we will now prove the strict inclusions among these sets
stated in the following theorem.

\begin{theorem}\label{thm:strict-inclusions}
  $\calI(\concake) \stackrel{\textrm{(a)}}{\subsetneq}
  \calI(\concake)^\N \stackrel{\textrm{(b)}}{\subsetneq}
  \calB(\concake) \stackrel{\textrm{(c)}}{\subsetneq}
  \calL(\concake) \stackrel{\textrm{(d)}}{\subsetneq}
  \mathfrak{P}(\concake)$.
\end{theorem}

\begin{proof}
  We start with proving~(a): $\calI(\concake) \subsetneq
  \calI(\concake)^\N$.  Obviously, $\calI(\concake) \subseteq \calI(\concake)^\N$
  is true, as every finite union of intervals is a countable union of
  intervals. To see that the two sets are not equal, look at $I =
  \bigcup_{i\in\N\cup\{0\}} [3\cdot 2^{-i-2},2^{-i}]$.  It is clear that
  $I \in \calI(\concake)^{\N}$ is true, as $I$ is a countable union of
  intervals. However, it holds that $I = [ \nicefrac{3}{4}, 1 ] \cup [
  \nicefrac{3}{8}, \nicefrac{1}{2} ] \cup \cdots$, i.e., $I$ cannot be
  written as a finite union of intervals, as all these subintervals are
  pairwise disjoint.  Hence, $I \notin \calI(\concake)$, so $\calI(\concake)
  \subsetneq \calI(\concake)^{\N}$, and we have shown~(a).

  In Lemma~\ref{lem:Borel-sigma-algebra} and
  Definition~\ref{def:Borel-sigma-algebra}, we have seen that
  $\calB(\concake) = \sigma(\calQ(\concake))$ where $\calQ(\concake)$ is the
  family of all intervals within~$\concake$. Since a $\sigma$-algebra is
  stable under (finite and countable) unions, we get $\calI(\concake)
  \subseteq \sigma(\calQ(\concake)) = \calB(\concake)$. Using again the
  stability of a $\sigma$-algebra under countable unions, we arrive at
  $\calI(\concake)^{\N} \subseteq \calB(\concake)$.

  Since, however, $\QQ \cap \concake \in \calB(\concake)$ is true, as $\QQ
  \cap \concake$ can be written as a countable union of intervals that each
  contain one element, it must hold that $\overline{\QQ \cap \concake} \in
  \calB(\concake)$ by the definition of a $\sigma$-algebra.  However, the
  irrational numbers $\overline{\QQ \cap \concake}$ in $\concake$ cannot be
  written as a countable union of intervals, since every interval
  containing more than one element immediately contains a rational
  number. Therefore, $\calI(\concake)^{\N}$ is not an algebra and
  $\calI(\concake)^{\N} \neq \calB(\concake)$ holds, proving~(b).

  The inclusion $\calB(\concake) \subseteq \calL(\concake)$ holds by
  definition, as all Borel sets are also Lebesgue sets.  However, there
  are Lebesgue sets that are not Borel sets: Observe
  that the cardinality of $\calL(\concake)$ is the cardinality of
  $\mathfrak{P}(\concake)$ (which is $2^{\mathfrak c} > \mathfrak c$),
  whereas there are only continuum-many (i.e., $\mathfrak c$, the
  cardinality of $\concake$) Borel sets (see~\citeR[Appendix~G,
  Corollary~G.7]{schi:b:measures}).  This proves~(c).  An alternative
  direct construction can be based on the Cantor function, also known as
  the \emph{devil's staircase} (see~\citeR[p.~153,
  Example~7.20]{schi:b:counterexamples}).

  Finally, the power set $\mathfrak{P}(\concake)$ trivially contains all
  other families of sets considered earlier. Nevertheless, there are
  sets in $\mathfrak{P}(\concake)$ that are not Lebesgue sets, for
  example the Vitali set that we introduced in Example~\ref{ex-vitali},
  so $\calL(\concake) \neq \mathfrak{P}(\concake)$, and we have~(d).
\end{proof}

\section{Discussion}
\label{chap:discussion}

Taking $\calP=\calI(\concake)$ as domain for a valuation $v$ and a
protocol involving a finite number of cuts is always possible; this
remains true for open-ended protocols that stop after a finite but
\emph{a priori} unknown number of steps.  If the protocol is infinite,
the naive choice $\calP = \calI(\concake)^\N$ is problematic, as
$\calI(\concake)^\N$ is not an algebra and thus does not even satisfy
the minimum requirements for~$\calP$ as described in the first
paragraph of Section~\ref{chap:preliminaries:finite}.

From a theoretical point of view, however, the choice $\calP =
\calI(\concake)$ may be unnecessarily restrictive, especially in light of
the fact that we also want to use infinite cake-cutting protocols.
Therefore, a larger set $\calP$ may be desirable, perhaps even larger
than $\calI(\concake)^{\N}$, which (as we have seen) has disqualified
itself.

We start our discussion by explicating why $\calP =
\mathfrak{P}(\concake)$ is a bad choice and we then provide arguments for
a better option, namely the Borel $\sigma$-algebra $\calP =
\calB(\concake)$.

\subsection{Taming $\boldsymbol{\mathfrak{P}(\concake)}$ with Exotic
  Valuations via Banach Limits}
\label{chap:solution:ultra-filter}

If one boldly desires to define valuation functions on the set
$\mathfrak{P}(\concake)$ of all subsets of the cake, it remains to be
shown that this indeed is possible.  We have seen that the commonly
used valuation functions represented via boxes, as depicted in
\cref{fig:motivation:common-valuation-function}, are not capable of
evaluating every piece of cake in $\mathfrak{P}(\concake)$. Hence, in
this section we aim to define a valuation function capable of
evaluating \emph{every} possible piece of cake in
$\mathfrak{P}(\concake)$.

Let us begin with a negative result.

\subsubsection{A Negative Result}

Using axiomatic set theory one can show that there cannot be a
valuation $v$ of the standard cake $\concake$ which
\begin{enumerate}
\item[a)] is defined on all of $\mathfrak{P}(\concake)$,
\item[b)] is $\sigma$-additive, and
\item[c)] is divisible, hence satisfies $v(\{x\})=0$ for any $x\in\concake$.
\end{enumerate}
 The requirements a)--c) are a consequence of a result by Ulam, and it
requires that the continuum hypothesis holds true, see the books by
\citeA[p.~26, Proposition~5.7]{oxt:b:measure-and-category} or
\citeA[pp.~132--3, Example~6.15]{schi:b:counterexamples}.

 We may relax on b), i.e., the $\sigma$-additivity, and look for
\textbf{finitely} additive valuations if we want to admit all pieces
of cake. Let us formally define a valuation function $\mu$ on
$\mathfrak{P}(\concake)$ satisfying the requirements $\mathrm{(M)},
\mathrm{(A)}$, and $\mathrm{(D)}$ from~\cref{def:valuation-function}.
To do so, in a first step, we must choose an arbitrary sequence
$(x_i)_{i \in \N}$ of pairwise distinct elements from $\concake$. For
every $A \subseteq \concake$, we define a mapping $f_A \colon \N \to
[0,1]$ with
\begin{align}\label{equ:rev1}
  n \mapsto f_{A}(n)= \frac{|A \cap  \{ x_1, \dots, x_n \}|}{n},
\end{align}
where $|B|$ denotes the cardinality of any set~$B$.  That is, $f_A(n)$
describes the relative frequency of the first $n$ elements of
$(x_i)_{i \in \N}$ being in~$A$.  For some sets $A$ the limit
$\lim_{n\to\infty} f_A(n)$ does exist, but it may not exist for other
sets~$A$. We can, however, use the Banach limits, that we will now
introduce.

\subsubsection{Banach Limits}
\label{chap:preliminaries:banach}

We will need a nonconstructive way to extend linear maps.  The key
result is the standard Hahn--Banach theorem, which is well-known from
functional analysis (see,
e.g.,~\citeR[Theorem~3.2]{rudin:b:functional}), so we need to go on a
quick excursion into functional analysis.

\begin{theorem}
  Assume that $(Y,\|\cdot\|)$ is a normed vector space and $L:M\to\RR$
  a linear functional, which is defined on a linear subspace $M\subseteq
  Y$ satisfying $|Lx|\leq \kappa\|x\|$ for all $x\in M$ with a universal
  constant $\kappa=\kappa_L\in (0,\infty)$. Then there is an extension
  $\hat L : Y\to\RR$ such that $\hat L$ is again linear and satisfies
  $|\hat Lx|\leq \kappa\|x\|$ for all $x\in Y$ with the same constant
  $\kappa=\kappa_L$ as before.
\end{theorem}

With a little more effort, but essentially the same proof, we can
replace the norm $\|x\|$ (respectively, $\kappa\|x\|$) by a general
sublinear map $p:Y\to\RR$. Sublinear means that
$p(\alpha x)=\alpha p(x)$ and $p(x+y)\leq p(x)+p(y)$ for all $x,y\in
Y$ and $\alpha\geq 0$. In this case, the extension of $Lx\leq p(x)$
satisfies $-p(-x)\leq \hat Lx \leq p(x)$.  Note that $p$ is only
positively homogeneous, i.e., it may happen that $-p(-x)\neq p(x)$.

The proof is nonconstructive and, at least for nonseparable spaces
$Y$, relies on the axiom of choice.

We will use the Hahn--Banach theorem for the space of bounded
sequences $\ell^\infty([0,\infty)) = \{x=(x_n)_{n\in\N}\subset
[0,\infty) \mid \|x\|_\infty < \infty\}$, where $\|x\|_\infty =
\sup_{n\in\N}x_n$ is the uniform norm.  Note that
$(\ell^\infty([0,\infty)),\|\cdot\|_\infty)$ is a nonseparable space.

A prime example of a bounded linear functional is the limit: Consider
those $x=(x_n)_{n\in\N}\in \ell^\infty([0,\infty))$ where
$L(x):=\lim_{n\to\infty}x_n = x$ exists in the usual sense. It is
common to write $c([0,\infty)) = \{x\in\ell^\infty([0,\infty)) \mid
\lim_{n\to\infty}x_n \text{\ exists}\}$. Clearly, $\lim_{n\to\infty}
x_n = \limsup_{n\to\infty} x_n\leq \sup_{n\in\N} x_n$, so that $L$ is
a bounded linear functional on $M=c([0,\infty))\subset
Y=\ell^\infty([0,\infty))$, and we can extend it to all of $Y$ as the
\textbf{Banach limit}, i.e.,
\begin{align*}
  \LIM_{n\to\infty} x_n
  :=
  \begin{cases}
    \lim_{n\to\infty} x_n &\text{if $x\in c([0,\infty))$},\\
    \hat{L}(x) &\text{if $x\in \ell^\infty([0,\infty))\setminus
      c([0,\infty))$}.
  \end{cases}
\end{align*}
Using the addition to the Hahn--Banach theorem with
$p(x):=\limsup_{n\to\infty} x_n$ and the observation that
$\lim_{n\to\infty}x_n$ exists if, and only if,
$\liminf_{n\to\infty}x_n=\limsup_{n\to\infty}x_n\in[0,\infty)$, we can
choose the extension $\hat L$ in such a way that
\begin{gather*}
    \liminf_{n\to\infty} x_n \leq \LIM_{n\to\infty} x_n \leq \limsup_{n\to\infty} x_n.
\end{gather*}

The construction of Banach limits is a typical application of the
Hahn--Banach extension theorem, hence the axiom of choice. The
appearance of these two concepts in this context is not an
accident. The seminal paper of \citeA{ban:j:probleme-de-la-mesure}
(see
also~\citeR[Chapter~II.\S1]{ban:b:theorie-des-operations-lineaires})
proves what we now call the \enquote{Hahn--Banach extension theorem
for linear functionals} in order to solve the \emph{probl\`eme de la
m\'esure} by \citeA[Chapter~VII.ii]{leb:b:leccons-sur-l-integration}
which asks for the existence of an additive, or $\sigma$-additive,
translation invariant measure on $\mathfrak{P}(\RR^n)$. The answer
depends on the dimension: In dimension $n\geq 3$, it is always
negative (because of the Banach--Tarski paradox), whereas in
dimensions $1$ and~$2$ it is negative if the measure is to be
$\sigma$-additive (because of Vitali-type constructions, cf.\
Example~\ref{ex-vitali}). More on this can be found in the books by
\citeA[Chapter~10]{wag:b:banach-tarski-paradox} and
\citeA[Example~7.31]{schi:b:counterexamples}.

There is a deep connection between the underlying group structure of
the space $\RR^n$ and Lebesgue's measure problem (this was discovered
by \citeR{neu:j:analytische-eigenschaften}).  Following M. M. Day, a
group $\mathbb G$ which allows for finitely additive,
(left-)translation invariant measures on all of $\mathfrak{P}(\mathbb
G)$ is nowadays called \textbf{amenable}~-- a pun combining the actual
meaning of the word (\enquote{nice, comfortable}) with its
pronunciation which reminds of \enquote{mean value} or measure. The
axiom of choice, which is needed for Hahn--Banach, can also be used to
construct extensions of measures defined on a sub-algebra $\calA_0$ of
an algebra $\calA$. It is known that this extendability, essentially,
is equivalent to the Hahn--Banach theorem (cf.~\citeR[Theorem~10.11 and
Corollary~13.6]{wag:b:banach-tarski-paradox}) describing its axiomatic
strength.

\subsubsection{From Banach Limits to Valuation Functions}

Having defined and discussed Banach limits, we will now use them to
construct,  based on the function $f_A$ defined in \eqref{equ:rev1},
a valuation function
\begin{align*}
  \mu \colon \mathfrak{P}(\concake) \to [0,1], \quad A \mapsto
  \LIM_{n \to \infty} f_A(n).
\end{align*}

It is clear that $\mu(A)$ is additive since $A\mapsto f_A(n)$ is
additive (for every fixed $n$) and both the limit and the Banach limit
are additive, so property $\mathrm{(A)}$ from
Definition~\ref{def:valuation-function} is satisfied. In the following
lemma we show that $\mu$ satisfies property~$\mathrm{(D)}$. At first
glance, this seems to contradict Corollary~\ref{cor-D-vs-sigma}. But
divisibility $\mathrm{(D)}$ involves the domain of the valuation, and
the proof of the lemma shows that we have almost no control on the set
$A_\alpha \subseteq A$ which achieves divisibility. That means, the
following phenomenon is symptomatic for having a \enquote{too big
domain.}

\begin{lemma}\label{lem:ultra-filter-valuation-function-satifies-D}
  For every $A \in \mathfrak{P}(\concake)$ with $\mu(A) > 0$ and every
  real number $\alpha \in [0,1]$, there exists a subset $A_\alpha
  \subseteq A$ in $\mathfrak{P}(\concake)$ such that $\mu(A_\alpha) =
  \alpha \mu(A)$.
\end{lemma}

\begin{proof}
  If $\mu(A)>0$ then $A$ must contain an infinite number of points of
  the underlying sequence, say $A \cap \{x_1, x_2, \dots \} = \{
  x_{i(1)}, x_{i(2)}, \dots \}$ for some increasing sequence
  $(i(k))_{k\in\N}$ of integers. By assumption, $A\cap \{x_1, x_2,
  \dots, x_n\} = \{x_{i(1)}, \dots, x_{i(m)} \mid i(m)\leq n\}$, and so
  $\mu(A) = \LIM_{n\to\infty}\frac 1n |\{x_{i(1)},\dots, x_{i(m)} \mid
  i(m)\leq n\}|$.
  We have to construct a set $B\in\mathfrak{P}(A)$ such that
  $\LIM_{n\to\infty} f_B(n)=\alpha\mu(A)$ for fixed $\alpha\in[0,1]$.

  The key observation in this proof is the fact that for any
   positive  rational number $\nicefrac{k}{n}$ with $k<n$, we have
  \begin{gather*}
  \text{ on the one hand}\quad
    \frac{k}{n+1} < \frac kn
  \quad\text{ and on the other hand}\quad
    \frac kn < \frac{k+1}{n+1},
  \end{gather*}
  i.e., the quantity $f_{B}(n) = \frac 1n
  |B\cap\{x_{i(1)},\dots,x_{i(m)} \mid i(m)\leq n\}|$ decreases if we
  jack up $n\to n+1$ and the numerator does not increase, i.e., if
  $i(m+1)>n+1$ or if $x_{i(m+1)} = x_{n+1}\not\in B$, and it increases
  if we jack up $n\to n+1$ and $x_{i(m+1)} = x_{n+1}\in B$.

  Fix $\alpha \in [0,1]$ and observe that we can assume that $0<\alpha
  < 1$: If $\alpha =0$, we take $B=\emptyset$, and for $\alpha = 1$, we
  use $B=\{x_{i(m)} \mid m\in\N\}$. For $\alpha\in (0,1)$, we use a
  recursive approach.

  Since $\mu(A) = \LIM_{n\to\infty}\frac 1n |\{x_{i(1)},\dots,
  x_{i(m)} \mid i(m)\leq n\}|$ and $0<\alpha<1$ we must have $\frac
  1{n(1)} |\{x_{i(1)},\dots, x_{i(m)} \mid i(m)\leq
  n(1)\}|\geq\alpha\mu(A)$ for some $n(1)\in\mathbb N$. Define
  $B_{n(1)}=\{x_{i(1)},\dots, x_{i(m)} \mid i(m)\leq n(1)\}$, and assume
  that we have already found a set $B_n$ such that
  $f_{B_n}(n)\geq\alpha\mu(A)$. Because of the observation at the
  beginning of the proof, the numbers
  \begin{align*}
    \ell_{n+1} &:= \min\left\{k>n \mid
                 f_{B_n}(k)\leq\alpha\mu(A)\right\}~\text{and}\\
    u_{n+1} &:= \min\left\{k>\ell_{n+1} \mid
              f_{A_k}(k)\geq\alpha\mu(A) \text{ for }A_k= B_n\cup\{
              x_{i(\ell_{n+1}+1)},\dots, x_{i(k)}\}\right\}
  \end{align*}
  are well-defined and satisfy $\ell_n < u_n < \ell_{n+1}<u_{n+1}$ and
  $\ell_{n+1}\to\infty$. Setting
  \begin{gather*}
    B_{n+1} = B_{n} \cup \{x_{i(\ell_{n+1}+1)},\dots, x_{i(u_{n+1})}\}
  \end{gather*}
  finishes the recursion, and we can define $B=\bigcup_{n\geq n(1)} B_n$.

  By construction, $|f_{B}(n)-\alpha\mu(A)|\leq \ell_{n+1}^{-1}$ holds
  for $n>n(1)$, completing this proof.
\end{proof}

We now provide a counterexample that shows that $\mu$ is not
$\sigma$-additive. To do so, we define $A_0 = \concake \setminus \{ x_i
\mid i \in \N \}$ and $A_i = \{x_i\}$ for $i \in \N$. Obviously, for
all $j \in \N\cup \{0\}$, it holds that $\mu(A_j) = 0$, while at the
same time we have
\begin{align*}
  \mu \left( \bigcup_{j \in \N\cup \{0\}} A_j \right) = \mu(\concake) = 1,
\end{align*}
which means that $\mu$ is not $\sigma$-additive.

The valuation function $\mu$ defined above may seem to be attractive
for cake-cutting. We can interpret the sequence $(x_i)_{i \in \N}$ as
countably many points which are used to evaluate arbitrary pieces of
the cake. However, there are multiple drawbacks. First of all, the
existence of a Banach limit is only guaranteed if one is willing to
accept the validity of the axiom of choice, as already mentioned in
\cref{chap:preliminaries:banach}.  Furthermore, until now no explicit
nontrivial example of a Banach limit is known.  Hence, we cannot
calculate $\mu(A)$ for $A \in \mathfrak{P}(\concake)$ if the ordinary
limit of $f_A(n)$ does not exist, as we do not know what the Banach
limit looks like.

Thus, although $\mu$ is theoretically capable of evaluating all pieces
of cake in $\calP = \mathfrak{P}(\concake)$, it is actually not useful
for our purposes.  Besides the previously listed mathematical
problems, there are also practical problems related to cake-cutting
itself. If we would use $\mu$ as a valid valuation function in
cake-cutting, all players would be obliged to precisely define a
countable sequence $(x_i)_{i \in \N}$ of pairwise distinct elements in
$\concake$ and some Banach limit they are using for their valuation
functions. When we think of common approaches and results in the
cake-cutting literature, this approach seems impractical and not
feasible to use.

Hence, finding practically usable valuation functions defined on
$\mathfrak{P}(\concake)$ seems to remain an open problem.  Nonetheless,
this section showed that defining more complex valuation functions
(compared to the valuation functions represented via boxes) does not
solve our initial problem on $\mathfrak{P}(\concake)$. Therefore, in the
next section we discuss an alternative solution, namely, reducing
$\calP$ in size from $\mathfrak{P}(\concake)$ to a smaller subfamily
contained in $\mathfrak{P}(\concake)$.

\subsection{Borel $\boldsymbol{\sigma}$-Algebra}

We recommend to use $\calP = \calB(\concake)$ as the most useful family
of all admissible pieces of cake.
As shown in \cref{thm:strict-inclusions}, the Borel
$\sigma$-algebra $\calB(\concake)$ (strictly) contains $\calI(\concake)$
as well as $\calI(\concake)^{\N}$, but is strictly smaller than
$\calL(\concake)$ and $\mathfrak{P}(\concake)$.

In general, the Borel $\sigma$-algebra can become quite large and
complicated if the base set is not countable, as is the case for
$\concake \subset \RR$. In particular, one needs transfinite
induction to \enquote{construct} all Borel sets.  This means that, in
general, we cannot construct a valuation $\mu$ on $\calB(\concake)$ by
explicitly assigning a value $\mu(A)$ to every element
$A\in\calB(\concake)$ nor give a recursive algorithm to
construct~$\mu(A)$, as the $\sigma$-algebra is simply too
large. Instead, one can describe the valuation on a suitable generator
of the $\sigma$-algebra and use Carath\'{e}odory's extension theorem,
stated previously as Theorem~\ref{thm:cara}.

Let us show here that the box-based valuation functions are
$\sigma$-additive valuations on $\calI(\concake)$ and that $\calI(\concake)$
is an algebra. In this case, we can use Carath\'eodory's extension
theorem to extend the valuation functions to measures on
$\sigma(\calI(\concake)) = \calB(\concake)$. Since the valuation functions
are  finite,  it follows that this extension is unique. Hence,
by providing a box-based valuation function, we obtain a unique
measure on $\calB(\concake)$. Thus $\calP = \calB(\concake)$ is a good
solution to our problem.

Let us formalize the box-based valuation functions. A box-based
valuation function $\mu$ partitions the complete cake $\concake$ into a
finite number of pairwise disjoint subintervals, where each
subinterval is allocated a finite number of boxes of equal height.  We
denote the set of all subintervals which $\mu$ uses by
\begin{align*}
  \calI_{\mu} = \{ I_1 = [a_1, b_1), \dots, I_{n-1}=[a_{n-1},b_{n-1}),
  I_n = [a_n, b_n] \},
\end{align*}
where $\bigcup_{i=1}^{n} I_i = \concake$ and we have $I_i \cap I_j =
\emptyset$ for all $i$ and~$j$, $1 \leq i < j \leq n$.  Furthermore,
denote by $\psi_i \in \N$, for $1 \leq i \leq n$, the number of boxes
allocated to an interval $I_i$ by $\mu$ and denote by $\psi_{\mu} =
\sum_{i=1}^{n} \psi_i$ the total number of boxes.  This gives the
following weight function
\begin{gather*}
    p(x)
    := \frac{1}{\psi_\mu} \sum_{i=1}^n \frac
    {\psi_i}{\lambda(I_i)}\textbf{1}_{I_i}(x)
    = \frac{1}{\psi_\mu} \sum_{i=1}^n \frac
    {\psi_i}{b_i-a_i}\textbf{1}_{I_i}(x).
\end{gather*}
Note that $p(x)$ is an  integrable  function, which is a
probability density, i.e., $\int_0^1 p(x)\,dx = 1$.  Since $\mu$ is a
Lebesgue measure with a weight, we cannot define it on all of
$\mathfrak{P}(\concake)$, but we may extend it easily onto $\calB(\concake)$
using integration: Define $\mu \colon \calB(\concake) \to [0,1]$ as
\begin{gather*}
    B\mapsto \mu(B) := \int_B p(x)\,dx.
\end{gather*}
In particular, if $B \in \calI(\concake)$ is a finite union of intervals
in~$\concake$, we see that
\begin{align*}
  B \cap I_i = \bigcup_{j=1}^{n(B,i)} \langle c_j^i,
  d_j^i\rangle,\quad i=1,2,\dots n
\end{align*}
for suitable $n(B,i)\in\N$, and
\begin{align*}
  \mu(B) = \frac{1}{\psi_{\mu}} \sum_{i=1}^{n} \left[
  \frac{\psi_i}{b_i - a_i}
  \sum_{j=1}^{n(B,i)} (d_j^i - c_j^i) \right].
\end{align*}

\begin{example}
  Referring back to the box-based valuation function $\nu$ from
  \cref{fig:motivation:common-valuation-function} on
  page~\pageref{fig:motivation:common-valuation-function}, we obtain
  \begin{align*}
    \calI_{\nu} = \{ I_1 = [0, \nicefrac{1}{6}),
      I_2 = [\nicefrac{1}{6}, \nicefrac{2}{6}), \dots,
      I_6 = [\nicefrac{5}{6}, 1] \}.
  \end{align*}

  Also, we have $\psi_1 = 2$, $\psi_2 = 1$, $\psi_3 = 5$, $\psi_4 =
  2$, $\psi_5 = 4$, $\psi_6 = 3$, and $\psi_{\nu} = 17$. For $B =
  [\nicefrac{1}{10},\nicefrac{1}{4}]$, we obtain
  \begin{align*}
    \nu(B) & = \frac{1}{\psi_{\nu}} \sum_{i=1}^{6} \left[
                    \frac{\psi_i}{b_i - a_i} \sum_{j = 1}^{n(B,i)}
                    (d_j^i - c_j^i) \right] \\
                  & = \frac{1}{17} \left(\frac{2}{\nicefrac{1}{6}} \cdot
                    \left(\nicefrac{1}{6} -  \nicefrac{1}{10}\right)
                    + \frac{1}{\nicefrac{1}{6}} \cdot
                    \left(\nicefrac{1}{5} - \nicefrac{1}{6}\right) +
                    \frac{5}{\nicefrac{1}{6}} \cdot 0 + \frac{2}{\nicefrac{1}{6}} \cdot 0
                    + \frac{4}{\nicefrac{1}{6}} \cdot 0 + \frac{3}{\nicefrac{1}{6}} \cdot
                    0\right) \\
                  & = \frac{ 1}{17}.
  \end{align*}
\end{example}

Summing up, $\calB(\concake)$ is recommended as a very good choice for
$\calP$, since this choice enables us to use box-based valuation
functions and their extensions as measures.  It also enables us to use
any probability density -- not only piecewise continuous densities --
to define a divisible valuation function on $\calB(\concake)$ which is
an absolutely continuous probability measure with respect to Lebesgue
measure.  Since Lebesgue null sets are subsets of Borel null sets, and
since an absolutely continuous valuation attaches value zero to any
Borel null set, a further extension to $\calL(\concake)$ is also
possible, but the enrichment by subsets of Borel null sets (which are
evaluated zero) has no additional benefit.

\section{Conclusion and Some Further Technical Remarks}
\label{chap:conclusion}

Among the questions we have tried to answer are:
\begin{enumerate}
\item Which subsets of $\concake$ should be considered as pieces of cake?
  Only finite unions of intervals or more general sets?
\item If valuation functions are considered as set-functions as
  studied in measure theory, should they be \emph{$\sigma$-additive} or
  only \emph{finitely additive}?
\end{enumerate}
A related interesting question is:
\begin{enumerate}\setcounter{enumi}{2}
\item Which continuity property should be used for a valuation?

  For the standard cake $\concake$, the natural choices are either
  \emph{divisibility} $\mathrm{(D)}$ or \emph{absolute continuity} with
  respect to Lebesgue measure, see p.~\pageref{ac-sc-d}. Obviously,
  absolute continuity implies continuity. There is a partial converse to
  this assertion: The notions of continuity and divisibility coincide
  (cf.\ Corollary~\ref{cor-D-vs-sigma}) and the distribution function
  $F_v(x)$ of a continuous valuation can be represented as a sum of the
  form $F_v(x) = \int_{0}^{x} f(t)\,dt + v^{\mathrm{sc}}([0,x])$; this
  means that it has an absolutely continuous part and a
  continuous-singular part, see the discussion in the paragraph on
  \enquote{General Measures} following Definition~\ref{def:measure}. For
  an abstract cake, one should replace divisibility $\mathrm{(D)}$ by
  the notion of sliceability, which is equivalent to condition
  $\mathrm{(DD)}$ by Theorem~\ref{theo-conti}, see
  Section~\ref{subsec:abstract} and
  \citeA{sch-sto:t:continuity-assumptions-in-cake-cutting}.
\end{enumerate}
While one can define the Dirac and counting measures for all sets in
$\mathfrak{P}(\abscake)$, there is no way to define a geometrically
sensible (and $\sigma$-additive [in dimensions one and two] or
finitely additive [in all higher dimensions]) notion of
\enquote{volume} for all sets~-- if we accept the validity of the
axiom of choice.  One can even show that the axiom of choice is
equivalent to the existence of non-measurable sets
(cf.~\citeR[p.~55]{cie:j:how-good-is-lebesgue-measure}).

Our findings result in concrete recommendations for cake-cutters. For
a finitely additive valuation $v$ on the standard cake $\concake$
(or indeed any one-dimensional cake) equipped with the algebra
$\calI(\concake)$ generated by the intervals, divisibility $\mathrm{(D)}$
is equivalent to atom-freeness or the continuity of the distribution
function $F_v(x) = v([0,x])$, cf.\ Corollary~\ref{cor-D-vs-sigma}. For
an abstract cake and a finitely additive valuation $v$, divisibility
$\mathrm{(D)}$ should be replaced by sliceability $\mathrm{(DD)}$,
which is equivalent to $v$ being atom-free; if $v$ is even
$\sigma$-additive, conditions $\mathrm{(D)}$ and $\mathrm{(DD)}$
coincide, see Theorem~\ref{theo-conti} and Corollary~\ref{cor-conti}.

 All of this breaks down, however, if we consider finitely additive
valuations on too big domains, say $\calP = \mathfrak{P}(\abscake)$: Even
for the standard cake there are divisible, finitely additive but not
$\sigma$-additive valuations, see
Lemma~\ref{lem:ultra-filter-valuation-function-satifies-D}.

We have also discussed in detail the measure-theoretic notions and
results that are relevant for the foundations of cake-cutting, for
both the standard cake and abstract cakes, including the notions of
$\sigma$-additivity, the Borel $\sigma$-algebra, and Carath\'eodory's
extension theorem (Theorem~\ref{thm:cara}). We emphasized the
importance of the Hahn--Banach theorem and the underlying axiom of
choice if one needs to evaluate arbitrary pieces of cake which are not
Borel or Lebesgue sets.

Banach, who can be seen as one of the founding fathers of the field of
cake-cutting,\footnote{Indeed, \citeA{ste:j:fair-division} presents
  the so-called last-diminisher procedure that is due to his students
  Banach and Knaster and guarantees a proportional division of the cake
  among any number of players.}  might perhaps have appreciated the
close connection between his work in measure theory and in
cake-cutting. For future work, we suggest to study which implications
our findings may have on existing or on yet-to-be-designed
cake-cutting algorithms.

To conclude, we have surveyed the existing rich literature on
cake-cutting algorithms and have identified the most commonly used
choices of sets consisting of what is allowed as pieces of cake.
After showing that these five most commonly used sets are distinct
from each other, we have discussed them in comparison.  In particular,
we have argued that $\mathfrak{P}(\abscake)$ is too general to define a
(practically or theoretically) useful valuation function on it.  And
finally, we have reasoned why we recommend the Borel $\sigma$-algebra
$\calB(\abscake)$ as a very good choice and how to construct, using
Carath\'eodory's extension theorem, a measure on $\calB(\abscake)$ that
cake-cutters can use to handle their box-based and even more general
valuation functions.

For a pragmatic approach to cake-cutting on the standard cake
$\concake$, the following five points are important:
\begin{enumerate}
\item[1.]  If one is interested in a fixed number of players and a
  fixed number of cuts, any additive valuation $v$ defined on the
  algebra of intervals $\calI(\concake)$ will do.
\item[2.]  If the players take rounds and if the protocol is
  open-ended
  (i.e., finite unbounded)
  or
  even
  infinite
  (recall the examples mentioned in Section~\ref{chap:introduction}
  and in Footnote~\ref{foo:infinite-protocols} of
  Section~\ref{chap:preliminaries:sigma}),

  the finite
  additivity of the valuation $v$ needs to be strengthened to
  $\sigma$-additivity, and the domain of the valuation should contain
  the Borel $\sigma$-algebra~$\calB(\concake)$~-- this is the
  smallest~$\sigma$-algebra containing $\calI(\concake)$.
\item[3.] If the valuation $v$ on $\calI(\concake)$ is divisible, measure
  theory guarantees that one is automatically in the situation described
  in item~2, i.e., the proper domain of (the extension of) $v$ is the
  Borel $\sigma$-algebra $\calB(\concake)$.
\item[4.] If one wants to extend the domain of the valuation $v$
  beyond $\calB(\concake)$, things become difficult: On the one hand, it is
  quite tricky to \enquote{construct} sensible valuations~-- unless we
  are happy with \enquote{rather simple} valuations like countable sums
  of point masses $v = \sum_{i\in\N} p_i \delta_{x_i}$, $\sum_{i\in\N}
  p_i = 1$, $(x_i)_{i\in\N}\subseteq [0,1]$, but these are obviously not
  divisible~-- and, on the other hand, they are not well-behaved,
  touching the very basis of axiomatic set theory.
\item[5.] The tools provided by measure theory are powerful enough to
  handle even abstract cakes.
\end{enumerate}

\begin{ack}
  We thank William S. Zwicker for helpful discussions.
\end{ack}

\begin{funding}
  This work was supported in part by grants of the Deutsche
  Forschungsgemeinschaft (DFG) RO-1202/14-2, RO-1202/21-1 and the DFG
  and Narodowe Centrum Nauki (NCN) joint initiative \enquote{Beethoven~3
    classic} SCHI-419/11-1, NCN 2018/31/G/ST1/02252.  The second author
  was a member of the PhD-programme \enquote{Online Participation,}
  supported by the North Rhine-Westphalian funding scheme
  \enquote{For\-schungs\-kollegs.}
\end{funding}

\appendix
\section{An Alternative Example Illustrating the Cantor Dust}
\label{sec:appendix}

  \begin{example}[Cantor dust; Cantor's ternary set]\label{ex:cantor-set-appendix}
    Write the elements $x\in\concake$ of the cake $\concake$ as
    ternary numbers, i.e., in the form
    \begin{gather*}
      x = \sum_{n\in\N} \frac{x_n}{3^n} \simeq 0.x_1x_2x_3\dots,
      \quad\text{where}\quad x_n\in\{0,1,2\},
    \end{gather*}
    and consider the set $C_{1/3}$ comprising all $x$ whose ternary
    expansion contains the digits \enquote{$0$} or \enquote{$2$} only. To
    enforce uniqueness, identify expressions of the form
    $0.*\!*\!*1000\ldots$ with $0.*\!*\!*0222\ldots$. The set $C_{1/3}$ is
    not countable since there is a bijection between $C_{1/3}$ and
    $\concake$: Take any $x=0.x_1x_2x_3\ldots\in C_{1/3}$ and read $\hat x:=
    0.\frac{x_1}2\frac{x_2}2\frac{x_3}2\ldots$ as \textbf{dyadic}
    expansion of an arbitrary element $\hat x\in \concake$.

    The set $C_{1/3}$ is the so-called \textbf{Cantor set} from
    Example~\ref{ex:cantor-set}. Think of its elements as \enquote{cream}
    pieces within the cake $\concake$, and imagine two players, taking turns
    in picking pieces of cake; for some reason (that their
    cardiologist elaborated on in detail) they have to avoid the cream
    altogether. For this, they are allowed to make two cuts, taking out an
    interval from the cake.\footnote{This is, of course, a non-standard
      cake-cutting protocol.}

    The optimal strategy is to take, in each round, the largest
    (necessarily open) interval between two cream pieces. From the triadic
    expansion, we see that, at each stage of the game, the maximum
    distance between two cream pieces is
    $0.\underbracket[.6pt]{*\!*\!*2}_n000\ldots -
    0.\underbracket[.6pt]{*\!*\!*0}_n222\ldots =
    0.\underbracket[.6pt]{0001}_n000\ldots \simeq 3^{-n}$, and this
    situation appears exactly $2^{n-1}$ times, since we have $2^{n-1}$
    choices for the leading $n-1$ digits denoted by the wildcard
    \enquote{$*\!*\!*$}~-- to wit, the pieces taken out are always the
    middle thirds of the largest remaining interval of cake:
    \begin{align*}
      A_0 = \concake
      &\xrightarrow[\qquad]{(*)}
        A_1 = A_0\setminus (\nicefrac{1}{3}, \nicefrac{2}{3}) =
        [0,\nicefrac{1}{3}] \cup [\nicefrac{2}{3},1]\\
      &\xrightarrow[\qquad]{(**)} A_2 = [0,\nicefrac{1}{9}] \cup
        [\nicefrac{2}{9},\nicefrac{1}{3}] \cup [\nicefrac{2}{3},
        \nicefrac{7}{9}] \cup [\nicefrac{8}{9}, 1].
    \end{align*}
    At the step marked $(*)$ Player 1 takes the first middle third, at
    the (double) step marked $(**)$ Player 2 and then Player 1 take the
    middle thirds of the remaining intervals, etc.

    If this procedure is repeated on and on, we remove countably many
    intervals from $\concake$ and end up with the \textbf{Cantor (ternary)
      set} $C_{1/3} = \bigcap_{n\in\N} A_n$ from
    Example~\ref{ex:cantor-set}, see
    \cref{fig:motivation:cantor-construction-sets-1}.

    We can use the triadic expansion also to assign a unique code to
    the removed piece: $I_{t_1t_2\dots t_{n}2}$ denotes the newly removed
    piece of cake at stage $n+1$ and the $t_1,\dots,t_{n}\in \{0,2\}$ mark
    the right-end point of the interval using the triadic expansion: $\sup
    I_{t_1t_2\dots t_{n}2} = \sum_1^{n} t_i 3^{-i} + 2\cdot
    3^{-n-1}$. This allows us to come up with a formula for the Cantor
    function from Example~\ref{ex:cantor-function}:
    \begin{gather}
        \text{on all of $I_{t_1t_2\dots t_{n}2}$ the function $V$ has
          the value\ \ } \sum_{i=1}^{n} \frac{t_i}{2} 2^{-i} + 3^{-n-1}.
    \end{gather}
    We refer to \cite[Sec.~2.5, 2.6]{schi:b:counterexamples} for a
    full discussion of this.
  \end{example}

\bibliographystyle{elsarticle-num-names}
\bibliography{sources}

\end{document}